\documentclass[10pt,twocolumn]{article}

\usepackage[a4paper,margin=2cm,columnsep=0.6cm]{geometry}
\usepackage{titlesec}
\usepackage{authblk}

\usepackage[T1]{fontenc}
\usepackage[utf8]{inputenc}
\usepackage{lmodern}
\usepackage{microtype}

\usepackage{amsmath,amssymb,amsfonts,amsthm}

\usepackage{booktabs}
\usepackage{multirow}
\usepackage{graphicx}
\usepackage{float}
\usepackage{caption}
\usepackage{tikz}
\usepackage{pgfplots}
\pgfplotsset{compat=1.18}

\usepackage{subcaption}
\usepackage{array}

\usepackage{algorithm}
\usepackage{algorithmic}

\usepackage{listings}
\usepackage[hidelinks,colorlinks=true,citecolor=blue!60!black,linkcolor=blue!60!black,urlcolor=blue!60!black]{hyperref}
\usepackage[numbers,sort&compress]{natbib}

\hypersetup{
  pdftitle={Aperiodic Structures Never Collapse: Fibonacci Hierarchies for Lossless Compression},
  pdfauthor={Roberto Tacconelli},
  pdfsubject={Lossless Data Compression},
  pdfkeywords={lossless compression, quasicrystal, Fibonacci tiling, arithmetic coding, Sturmian sequences, aperiodic hierarchy},
  pdfusetitle=false,
}

\usepackage{xcolor}
\usepackage{enumitem}
\setlist{nosep,leftmargin=*}
\titleformat{\section}{\large\bfseries}{\thesection.}{0.5em}{}
\titleformat{\subsection}{\normalsize\bfseries}{\thesubsection}{0.5em}{}
\titleformat{\subsubsection}{\normalsize\itshape}{\thesubsubsection}{0.5em}{}

\newtheorem{theorem}{Theorem}

\newtheorem{corollary}[theorem]{Corollary}

\theoremstyle{definition}
\newtheorem{definition}{Definition}
\newtheorem{remark}{Remark}

\newcommand{\qtc}{\textsc{Quasicryth}}

\renewcommand{\phi}{\varphi}    

\title{\textbf{Aperiodic Structures Never Collapse:\\
Fibonacci Hierarchies for Lossless Compression}}

\author[1]{Roberto Tacconelli}
\affil[1]{Independent Researcher\\
\texttt{tacconelli.rob@gmail.com}}

\date{}

\begin{document}
\maketitle
\thispagestyle{empty}

\begin{abstract}
Hierarchical dictionary methods lose effectiveness at deep scales
when their underlying structure collapses after finitely many levels.
We prove an \emph{Aperiodic Hierarchy Advantage} for Fibonacci
quasicrystal tilings: unlike periodic tilings, the Fibonacci hierarchy
never collapses, preserving non-zero $n$-gram lookup positions at every
depth and maintaining scale-invariant reuse capacity.
From this structure we derive constant potential word coverage
$W\phi/\sqrt{5}$, maximal Sturmian codebook coverage efficiency
$C_m/(F_m{+}1)$, bounded parsing overhead $1/\phi$ bits per word,
and lower coding entropy than comparable periodic hierarchies on
long-range-dependent sources.

Our analysis is supported by quantitative results on golden
compensation, Sturmian minimality, entropy, and redundancy.
The \emph{Golden Compensation theorem} shows that exponential decay
in position count is exactly balanced by exponential growth in phrase
length, making reuse capacity invariant across scales.
The \emph{Sturmian Codebook Efficiency theorem} links the minimal
complexity law $p(n)=n{+}1$ to maximal coverage efficiency, while
the redundancy analysis shows super-exponential decay
$O(e^{-\phi^m/\lambda})$ for the Fibonacci hierarchy, in contrast to
the finite-depth lock-in of periodic tilings.

We validate these results with
\qtc{},\footnote{\textsc{Quasicryth}: portmanteau of
\emph{quasicrystal} and \emph{compression}; the name also reflects
the internal structure ---
\textbf{quasi}c\textbf{ry}stalline \textbf{t}iling \textbf{h}ierarchy.
Source code: \url{https://github.com/robtacconelli/quasicryth}.}
a lossless text compressor built on a ten-level Fibonacci hierarchy
with phrase lengths $\{2,3,5,8,13,21,34,55,89,144\}$ words.
Version~5.6 extends the tiling engine with 36 multi-structure tilings
(12~golden-ratio Fibonacci phases plus 6~original non-golden irrational
tilings and 18~optimized additions discovered by iterative greedy
alpha search, including the far-out $\alpha=0.502$), scanning internally
to select the best tiling per block;
no tiling parameter appears in the compressed header.
In controlled A/B experiments with identical codebooks, the aperiodic
advantage over Period-5 grows from \textbf{36{,}243\,B} at 3\,MB to
\textbf{11{,}089{,}469\,B} at 1\,GB, explained by the activation of
deeper hierarchy levels.
On enwik9 (1\,GB), \qtc{} achieves \textbf{225{,}918{,}349\,B
(22.59\%)} in multi-structure mode, with 20{,}735{,}733\,B saved by
the Fibonacci tiling relative to no tiling.
\end{abstract}


\section{Introduction}
\label{sec:intro}

Shannon's equivalence of compression and prediction~\citep{shannon1948}
implies that superior structural models enable superior codes.
Classical lossless compressors exploit structure at the byte
level~\citep{ziv1977, huffman1952, witten1987} or through
Burrows--Wheeler block transformations~\citep{burrows1994}.
Word-level compressors~\citep{Horspool1992} apply frequency-ranked
codebooks, trading byte granularity for semantic structure.
In all cases, the \emph{parsing strategy} --- which sequence lengths to
attempt at each position --- is fixed or locally adaptive.

\qtc{} introduces a fundamentally different paradigm: the parsing
strategy itself is determined by a \emph{one-dimensional quasicrystal}.
The Fibonacci tiling $\sigma: L \to LS,\; S \to L$ generates an
aperiodic binary sequence in which the position and scale of every
super-tile are geometrically determined, providing $n$-gram lookup
positions at phrase lengths $\{2,3,5,8,13,21,34,55,89,144\}$ words
simultaneously.
The decisive property is that this hierarchy \emph{never collapses}:
unlike any periodic tiling, the Fibonacci quasicrystal sustains both
tile types at every scale, because the golden ratio $\phi$ is a
Pisot-Vijayaraghavan number and the tiling is Sturmian.

This paper makes the following contributions:

\begin{enumerate}
\item \textbf{Provable hierarchy non-collapse.}
A formal proof (Section~\ref{sec:theory}) that the Fibonacci
substitution hierarchy extends indefinitely, while all periodic tilings
with period $p$ collapse within $\lceil\log_\phi(p)\rceil$ levels.
The proof uses the Perron-Frobenius eigenvector analysis of the
substitution matrix, the PV property of $\phi$, and Weyl's
equidistribution theorem~\citep{weyl1910}.

\item \textbf{Sturmian structure as compression engine.}
The Fibonacci word is the canonical Sturmian sequence~\citep{morse1940}:
aperiodic, 1-balanced, and of minimal factor complexity $n{+}1$.
We show how these three properties translate directly into
compression benefits: balance ensures uniform codebook coverage,
aperiodicity prevents hierarchy collapse, and minimal complexity
maximises codebook entry reuse --- formalised as the Sturmian
Codebook Efficiency theorem (Theorem~\ref{thm:sturmian_eff}) and the
Goldilocks corollary (Corollary~\ref{cor:goldilocks}).

\item \textbf{Deep hierarchy with ten levels.}
\qtc{} v5.6 implements the full Fibonacci substitution hierarchy to
depth $k = 9$, providing lookup positions for phrases of up to 144
words.
At enwik9 scale (298\,M words), the 89-gram and 144-gram levels
activate for the first time, contributing 5{,}369 and 2{,}026 deep hits
respectively --- positions entirely unavailable to any periodic tiling.

\item \textbf{Aperiodic Hierarchy Advantage Theorem
(Theorem~\ref{thm:main}).}
Our central result (§\ref{subsec:main_result}) characterises the
Fibonacci quasicrystal as the only binary tiling, within the class of infinite aperiodic tilings, that preserves
dictionary reuse at \emph{all} hierarchy depths simultaneously.
Non-collapse, scale-invariant coverage, maximum codebook efficiency,
bounded overhead, and a strict coding entropy advantage over any
periodic alternative are all provable from first principles.
Eight supporting theorems (§\ref{subsec:matrix}--\ref{subsec:compression_bounds})
provide the quantitative details:
\begin{itemize}
  \item \emph{Sturmian Codebook Efficiency}
  (Theorem~\ref{thm:sturmian_eff}):
  the Fibonacci tiling has exactly $F_m{+}1$ distinct tile-type
  patterns at level $m$, the minimum for any aperiodic sequence.
  A codebook of $C_m$ entries achieves coverage efficiency
  $C_m/(F_m{+}1)$, provably maximum among all aperiodic binary tilings
  of the same depth.
  The Fibonacci tiling is characterised as the only aperiodic sequence satisfying both
  non-collapse and maximum efficiency at every level
  (Corollary~\ref{cor:goldilocks}).
  \item \emph{Golden Compensation} (Theorem~\ref{thm:golden_comp}):
  the potential word coverage $C(m)=P(m)\cdot F_m \to W\phi/\sqrt{5}$
  is the same at every hierarchy level.
  All structural variation is carried by the codebook hit rate $r(m)$
  alone; the positional and phrase-length factors cancel exactly via
  $\phi^2=\phi+1$ and Binet's formula.
  \item \emph{Activation Threshold} (Theorem~\ref{thm:activation}):
  hierarchy level $m$ first contributes when $W \geq W^*_m =
  T_m\phi^{m-1}/r_m$, giving a closed-form prediction for the corpus
  scale at which each deep level activates.
  \item \emph{Piecewise-Linear Advantage} (Theorem~\ref{thm:piecewise}):
  the aperiodic advantage $A(W)$ is convex and piecewise-linear in $W$,
  with slope strictly increasing at each $W^*_m$.
  The 33$\times$ advantage jump from 100\,MB to 1\,GB is the discrete
  signature of two new linear terms activating, not a superlinear
  regime of existing terms.
  \item \emph{Strict Coding Entropy} (Theorem~\ref{thm:entropy_ineq}):
  for any source with long-range phrase dependencies ($m$-LRPD),
  the per-word coding entropy of the Fibonacci hierarchy is
  \emph{strictly} lower than that of any periodic tiling collapsed
  before depth $m$ --- a fully information-theoretic inequality,
  not just a combinatorial one.
  Natural language corpora are $m$-LRPD for $m$ corresponding to
  phrase lengths up to at least 144 words
  (Corollary~\ref{cor:nl_lrpd}).
  \item \emph{Fibonacci Redundancy Bound}
  (Theorem~\ref{thm:redundancy}):
  for exponentially mixing sources, the coding redundancy at
  Fibonacci level $m$ decays \emph{super-exponentially} as
  $O(e^{-\phi^m/(\lambda\sqrt{5})})$, while any periodic tiling
  is locked at fixed redundancy $\Omega(e^{-F_{m^*}/\lambda})$.
  The redundancy ratio $\to 0$ as $m\to\infty$.
  \item \emph{Exponential Dictionary Efficiency}
  (Theorem~\ref{thm:dict_efficiency}):
  the per-entry compression gain $E_m = F_m\bar{h} - \log_2 C_m$
  grows as $\Omega(\phi^m)$; total dictionary value for Fibonacci
  grows exponentially in $k_{\max}$, versus a fixed constant for
  any periodic tiling.
  \item \emph{Convergent Flag Overhead} (Theorem~\ref{thm:flag_convergence}):
  the per-word entropy of the hit/miss flag stream for the Fibonacci
  tiling is bounded above by $1/\phi \approx 0.618$ bits/word,
  regardless of hierarchy depth.
  Combined with the unbounded growth of $h_\text{saved}$, this implies
  that the net per-word compression efficiency of the Fibonacci tiling
  exceeds that of any periodic alternative for all sufficiently large
  $W$ (Corollary~\ref{cor:h_dominance}).
\end{itemize}

\item \textbf{Quantified aperiodic advantage at scale.}
Controlled A/B experiments: at 3\,MB the
Fibonacci tiling saves 36{,}243\,B over Period-5; at 1\,GB this grows
to 11{,}089{,}469\,B.
The 89-gram and 144-gram levels activate at 1\,GB (enwik9),
contributing 5{,}369 and 2{,}026 deep hits --- positions structurally
unavailable to any periodic tiling.

\item \textbf{Multi-structure tiling engine.}
Version~5.6 extends the Fibonacci-only design to 36 tilings:
12~golden-ratio Fibonacci phases, 6~original non-golden irrational tilings,
and 18~optimized additions discovered by iterative greedy alpha search
(including $\alpha=0.502$, far below the golden ratio, which contributes
massive trigram/5-gram coverage).
The encoder scans all 36 candidates internally and selects the best
tiling per block; the decoder regenerates the complete tiling
deterministically without any tiling parameter in the compressed header.

\item \textbf{Competitive with established compressors.}
\qtc{} now surpasses bzip2 and approaches xz-level compression on
large corpora, demonstrating that quasicrystalline tiling hierarchies
are a practical compression engine, not merely a theoretical curiosity.

\item \textbf{Separate LZMA escape stream.}
Out-of-vocabulary words are segregated into a parallel stream
compressed with LZMA, achieving strong byte-level compression
versus inline arithmetic coding.
The quasicrystalline tiling governs which words escape, maintaining
the structural integrity of both streams.
\end{enumerate}

The remainder of the paper is organised as follows.
Section~\ref{sec:related} surveys related work.
Section~\ref{sec:method} describes the algorithm pipeline.
Section~\ref{sec:theory} develops the formal theoretical analysis.
Sections~\ref{sec:experiments}--\ref{sec:results} present the
experimental setup and results.
Section~\ref{sec:ablation} reports the A/B ablation study.
Section~\ref{sec:conclusion} concludes with future directions.

\section{Related Work}
\label{sec:related}

\subsection{Classical Lossless Compression}

Huffman coding~\citep{huffman1952} assigns variable-length binary codes
by symbol frequency; arithmetic coding~\citep{witten1987} removes the
integer-bit constraint, approaching the entropy bound to within a
fraction of a bit.
LZ77~\citep{ziv1977} and its derivatives (gzip/DEFLATE, LZMA/xz,
Zstandard) find repeated substrings within a sliding window.
Bzip2~\citep{seward1996} applies the Burrows-Wheeler
transform~\citep{burrows1994} followed by move-to-front and Huffman
coding, achieving strong compression on structured text.
Prediction by Partial Matching (PPM)~\citep{cleary1984} builds
high-order adaptive context models with arithmetic coding.
These approaches operate at the byte level and exploit
\emph{local} repetitions.
\qtc{} is a word-level compressor and exploits \emph{hierarchical}
phrase structure rather than byte-level redundancy.

\subsection{Word-Level and Phrase-Level Compression}

Word-based models~\citep{Horspool1992, moffat1992} replace
byte alphabets with word vocabularies, achieving better compression on
natural language by exploiting morphological and lexical structure.
LZW-style online dictionaries build phrase codebooks
adaptively~\citep{welch1984}.
Static frequency-ranked codebooks are used in specialized formats for
natural language text.
\qtc{} uses static multi-level codebooks built from the full input,
but the \emph{parsing} of the input into phrases is determined by the
quasicrystalline tiling rather than by greedy matching or fixed-length
frames.

\subsection{Quasicrystals and Aperiodic Tilings}

Quasicrystals were discovered experimentally in 1984 by
Shechtman et al.~\citep{shechtman1984} and have since been studied
extensively in mathematics~\citep{debruijn1981, senechal1995}.
The one-dimensional Fibonacci tiling is the canonical example of a
cut-and-project quasicrystal~\citep{duneau1985}: generated by
projecting a two-dimensional integer lattice onto the line
$y = \phi x$, it produces an aperiodic sequence of two tile types
(L and S) with a precisely irrational frequency ratio $\phi : 1$.

Substitution systems and their combinatorial properties are surveyed
in~\citep{allouche2003, lothaire2002}.
The Fibonacci word as the canonical Sturmian sequence is treated
in~\citep{morse1940, coven1973, berstel1995}.
The Three-Distance (Steinhaus) theorem is due to
Steinhaus~\citep{steinhaus1950} and Sós~\citep{sos1958}.
Pisot-Vijayaraghavan numbers and their role in substitution systems are
discussed in~\citep{pisot1938, queffelec2010}.
To our knowledge, \qtc{} is the first compressor to use a
quasicrystalline substitution hierarchy as the primary parsing
mechanism, and to prove that this hierarchy provides a structural
compression advantage over periodic alternatives with bounded period.

\subsection{Context Models and Hierarchy}

The PAQ family~\citep{mahoney2005} blends hundreds of adaptive
bit-level context models, achieving state-of-the-art compression on
structured data.
\qtc{} takes a different approach: rather than many parallel context
models, it uses a single deterministic quasicrystalline structure to
expose a deep hierarchy of phrase-level contexts.
The hierarchy context (3 bits per tile, free from the tiling) yields
8 specialised arithmetic coding sub-models per tile type at zero
bitstream cost.

\section{Method}
\label{sec:method}

\qtc{} processes input text through an eight-step pipeline.

\subsection{Step~1: Case Separation}

The input byte stream is tokenised into words and punctuation tokens.
Each token is classified as lowercase, titlecase, or uppercase,
yielding a 3-symbol case flag stream encoded separately with a
24-bit-precision adaptive arithmetic coder with order-2 context
(9 contexts from previous two case flags, variable-alphabet
Fenwick-tree models).
The remaining pipeline operates on the fully lowercased token stream.
On enwik9, case coding contributes 20{,}397{,}073\,B
($\approx$2.04\% of original), a necessary overhead for exact roundtrip.

\subsection{Step~2: Word Tokenisation}

The lowercased stream is split into word tokens: each token is either a
maximal alphabetic run (with trailing whitespace absorbed) or a single
non-alphabetic byte with trailing whitespace.
This produces a sequence of $W$ word tokens.
On enwik9, $W = 298{,}263{,}298$.

\subsection{Step~3: Multi-Level Codebook Construction}

Eleven frequency-ranked codebooks are built from the word token
sequence, corresponding to phrase lengths
$\{1,2,3,5,8,13,21,34,55,89,144\}$ --- the first eleven Fibonacci
numbers.
Table~\ref{tab:codebooks} shows the codebook sizes.
All n-gram codebooks are filtered to entries where every constituent
word appears in the unigram codebook.
For n-grams with $n \geq 13$, frequency counters are periodically
pruned (singletons removed every 500\,K entries) to prevent
out-of-memory on large inputs.
The serialised codebook is LZMA-compressed; at enwik9 scale this
occupies only 533{,}680\,B.

\begin{table}[ht]
\centering
\caption{Codebook sizes (number of entries) by input size tier.}
\label{tab:codebooks}
\small
\begin{tabular}{@{}lrrrrr@{}}
\toprule
\textbf{Level} & \textbf{Phrase} & \multicolumn{3}{c}{\textbf{Input word count}} \\
\cmidrule(lr){3-5}
 & \textbf{len} & $<$500K & $<$2M & $\geq$10M \\
\midrule
Unigram  & 1   & 8{,}000  & 16{,}000 & 64{,}000 \\
Bigram   & 2   & 4{,}000  & 8{,}000  & 32{,}000 \\
Trigram  & 3   & 2{,}000  & 4{,}000  & 32{,}000 \\
5-gram   & 5   & 1{,}000  & 2{,}000  & 16{,}000 \\
8-gram   & 8   &   500    & 1{,}000  & 4{,}000  \\
13-gram  & 13  &   300    & 1{,}000  & 4{,}000  \\
21-gram  & 21  &   100    &   500    & 2{,}000  \\
34-gram  & 34  &    50    &   200    & 2{,}000  \\
55-gram  & 55  &     0    &   100    & 1{,}000  \\
89-gram  & 89  &     0    &     0    & 1{,}000  \\
144-gram & 144 &     0    &     0    &   500    \\
\bottomrule
\end{tabular}
\end{table}

Index encoding uses variable-alphabet adaptive arithmetic coding
with Fenwick-tree acceleration, supporting codebook sizes up to
64{,}000 entries per level.

\subsection{Step~4: Quasicrystalline Word-Level Tiling}

v5.6 employs 36 aperiodic tilings from multiple families of irrational
numbers: 12 golden-ratio ($\alpha = 1/\phi$) tilings with phases
equally spaced by $1/\phi$, plus 6~original non-golden tilings
(2 each from $\alpha = \sqrt{58}-7$, noble-5, and $\sqrt{13}-3$),
and 18~optimized additions discovered by iterative greedy alpha
optimization.
The greedy search evaluates candidate $\alpha$ values by measuring
the marginal deep-position gain over existing tilings, then commits
the best candidate and repeats.
This process discovered far-out irrationals such as $\alpha = 0.502$
(well below the golden ratio $1/\phi \approx 0.618$), which provides
massive trigram and 5-gram coverage at positions complementary to all
golden-ratio tilings.
All 36 tilings produce quasicrystalline L/S sequences via the
cut-and-project method.

For tile position $k$, irrational slope $\alpha$, and phase
$\theta \in [0,1)$:
\begin{equation}
\text{tile}(k) =
\begin{cases}
L & \text{if } \lfloor(k+1)\alpha+\theta\rfloor
             - \lfloor k\alpha+\theta\rfloor = 1 \\
S & \text{otherwise}
\end{cases}
\label{eq:tiling}
\end{equation}
When $\alpha = 1/\phi$ (golden ratio), this reduces to the classical
Fibonacci quasicrystal.
Each $L$ tile \emph{consumes 2 words} (bigram lookup); each $S$ tile
consumes 1 word (unigram lookup).
The quasicrystalline matching rule --- no two $S$ tiles adjacent ---
is enforced by merging any resulting $SS$ pair into an $L$ tile.

A command-line flag \texttt{-f} restricts to the 12 golden-ratio
tilings only (Fibonacci-only mode) for controlled A/B comparison.

The compressor evaluates all 36 tilings, selecting the tiling
that maximises the scoring function:
\begin{equation}
\text{score} = \sum_k b_{\ell(k)}
\end{equation}
where $\ell(k)$ is the deepest hierarchy level applicable at tile $k$
and $b_\ell$ are exponentially increasing bonuses:
$b_0 = 3$ (unigram), $b_1 = 10$ (bigram),
$b_2 = 20$ (trigram), $b_3 = 50$, $b_4 = 100$, $b_5 = 200$,
$b_6 = 400$, $b_7 = 800$, $b_8 = 1{,}600$, $b_9 = 3{,}200$,
$b_{10} = 6{,}400$ (144-gram).

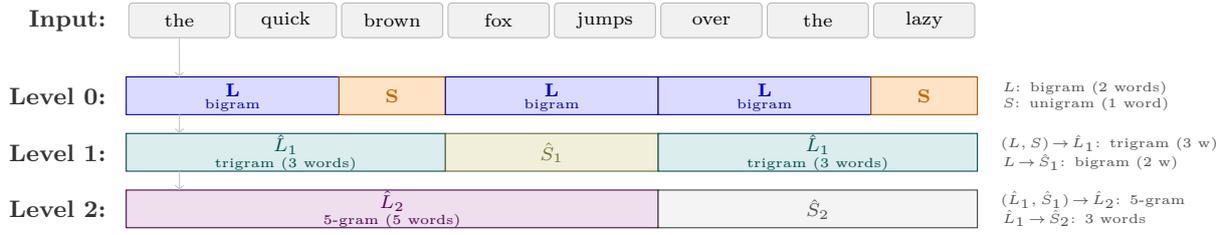
\begin{figure*}[t]
\centering
\begin{tikzpicture}[
  x=1.4cm, y=1cm,
  wbox/.style={draw=gray!55, fill=gray!11, rounded corners=1.5pt,
               minimum width=1.33cm, minimum height=0.46cm,
               font=\scriptsize, align=center},
]
  \foreach \i/\w in {0/the, 1/quick, 2/brown, 3/fox, 4/jumps, 5/over, 6/the, 7/lazy}
    \node[wbox] at (\i+0.5, 0.2) {\w};
  \node[anchor=east, font=\small\bfseries, text=black!80] at (-0.1, 0.2) {Input:};

  \draw[draw=blue!60!black, fill=blue!14] (0,-0.55) rectangle (2,-1.05);
  \node[font=\scriptsize\bfseries, text=blue!75!black] at (1,-0.72) {L};
  \node[font=\tiny, text=blue!60!black] at (1,-0.92) {bigram};
  \draw[draw=orange!65!black, fill=orange!22] (2,-0.55) rectangle (3,-1.05);
  \node[font=\scriptsize\bfseries, text=orange!80!black] at (2.5,-0.8) {S};
  \draw[draw=blue!60!black, fill=blue!14] (3,-0.55) rectangle (5,-1.05);
  \node[font=\scriptsize\bfseries, text=blue!75!black] at (4,-0.72) {L};
  \node[font=\tiny, text=blue!60!black] at (4,-0.92) {bigram};
  \draw[draw=blue!60!black, fill=blue!14] (5,-0.55) rectangle (7,-1.05);
  \node[font=\scriptsize\bfseries, text=blue!75!black] at (6,-0.72) {L};
  \node[font=\tiny, text=blue!60!black] at (6,-0.92) {bigram};
  \draw[draw=orange!65!black, fill=orange!22] (7,-0.55) rectangle (8,-1.05);
  \node[font=\scriptsize\bfseries, text=orange!80!black] at (7.5,-0.8) {S};
  \node[anchor=east, font=\small\bfseries, text=black!80] at (-0.1,-0.8) {Level~0:};

  \draw[draw=teal!55!black, fill=teal!14] (0,-1.3) rectangle (3,-1.8);
  \node[align=center, font=\scriptsize] at (1.5,-1.55) {%
    \textcolor{teal!65!black}{$\hat{L}_1$}\\[-2pt]%
    {\tiny\textcolor{teal!55!black}{trigram (3 words)}}};
  \draw[draw=olive!55!black, fill=olive!12] (3,-1.3) rectangle (5,-1.8);
  \node[font=\scriptsize, text=olive!65!black] at (4,-1.55) {$\hat{S}_1$};
  \draw[draw=teal!55!black, fill=teal!14] (5,-1.3) rectangle (8,-1.8);
  \node[align=center, font=\scriptsize] at (6.5,-1.55) {%
    \textcolor{teal!65!black}{$\hat{L}_1$}\\[-2pt]%
    {\tiny\textcolor{teal!55!black}{trigram (3 words)}}};
  \node[anchor=east, font=\small\bfseries, text=black!80] at (-0.1,-1.55) {Level~1:};

  \draw[draw=violet!55!black, fill=violet!13] (0,-2.05) rectangle (5,-2.55);
  \node[align=center, font=\scriptsize] at (2.5,-2.3) {%
    \textcolor{violet!65!black}{$\hat{L}_2$}\\[-2pt]%
    {\tiny\textcolor{violet!55!black}{5-gram (5 words)}}};
  \draw[draw=gray!50!black, fill=gray!9] (5,-2.05) rectangle (8,-2.55);
  \node[font=\scriptsize, text=gray!60!black] at (6.5,-2.3) {$\hat{S}_2$};
  \node[anchor=east, font=\small\bfseries, text=black!80] at (-0.1,-2.3) {Level~2:};

  \draw[->, thin, gray!45] (0.5,-0.03) -- (0.5,-0.55);   
  \draw[->, thin, gray!45] (0.5,-1.05) -- (0.5,-1.3);    
  \draw[->, thin, gray!45] (0.5,-1.8)  -- (0.5,-2.05);   

  \node[anchor=west, font=\tiny, text=black!70, align=left, text width=3.0cm] at (8.15,-0.8)
    {$L$: bigram (2 words)\\ $S$: unigram (1 word)};
  \node[anchor=west, font=\tiny, text=black!70, align=left, text width=3.0cm] at (8.15,-1.55)
    {$(L,S)\!\to\!\hat{L}_1$: trigram (3 w)\\ $L\!\to\!\hat{S}_1$: bigram (2 w)};
  \node[anchor=west, font=\tiny, text=black!70, align=left, text width=3.0cm] at (8.15,-2.3)
    {$(\hat{L}_1,\hat{S}_1)\!\to\!\hat{L}_2$: 5-gram\\ $\hat{L}_1\!\to\!\hat{S}_2$: 3 words};
\end{tikzpicture}
\caption{Concrete example of \qtc{} tiling on the phrase
``the quick brown fox jumps over the lazy'' (8\,words).
\textbf{Level~0:} the Fibonacci cut-and-project rule assigns
\textbf{L}~tiles (2-word bigram positions, blue) and
\textbf{S}~tiles (1-word unigram positions, orange), yielding
the sequence $L\;S\;L\;L\;S$ for these 8\,words.
\textbf{Level~1:} each $(L,S)$~pair is deflated into
$\hat{L}_1$ (trigram position, teal); isolated $L$~tiles become
$\hat{S}_1$.
\textbf{Level~2:} the pair $(\hat{L}_1,\hat{S}_1)$ yields
$\hat{L}_2$ (5-gram position, violet).
Arrows trace the lookup cascade at word~0 (``the''): the encoder
first attempts the 5-gram lookup
\emph{(``the quick brown fox jumps'')}, then trigram
\emph{(``the quick brown'')}, then bigram \emph{(``the quick'')},
falling back to unigram only if all longer lookups miss.
Word~5 (``over'') reaches a trigram entry point but no 5-gram.}
\label{fig:tiling}
\end{figure*}

\subsection{Step~5: Deep Substitution Hierarchy Construction}

The inverse Fibonacci substitution (deflation) is applied iteratively
to build a multi-scale hierarchy:
\begin{equation}
\sigma^{-1}: (L, S) \to \text{super-}L, \quad
L \text{ (isolated)} \to \text{super-}S
\label{eq:deflation}
\end{equation}
At each deflation level $k$, the super-L tile at that level spans
$F_{k+2}$ words, where $F_n$ is the $n$-th Fibonacci number.
The function \texttt{detect\_deep\_positions()} traces each $L$ tile
upward through the hierarchy: a tile at position $t_i$ can attempt a
level-$k$ n-gram lookup if and only if it is the leftmost child of a
super-L at every level up to $k$, and the span covers exactly $F_{k+2}$
words.
Deep matches from all 36 tilings are collected into position-indexed
arrays, then a greedy non-overlapping selection (deepest match wins)
assigns each word position to its optimal encoding level.

\begin{figure}[H]
\centering
\caption{Deep hierarchy position detection.
An $L$ tile at level~0 may qualify as the entry point of super-L
supertiles at levels 1--9, enabling lookups in the trigram through
144-gram codebooks respectively.
The qualifying condition is verified in a single upward pass through
the pre-computed hierarchy array.}
\label{fig:hierarchy}
\end{figure}

\subsection{Step~6: Multi-Level Adaptive Arithmetic Coding}

Encoding proceeds event-by-event over the greedy-selected assignment.
Five interacting mechanisms produce the final bitstream:

\paragraph{Order-2 level context.}
Level decisions use an order-2 context model conditioned on
$(\mathit{prev\_level}, \mathit{prev\_prev\_level})$, yielding
$12 \times 12 = 144$ specialised 12-symbol adaptive models.

\paragraph{Context-conditioned indices.}
Index models are conditioned on the previous level, giving 12 variants
per codebook level.

\paragraph{Recency cache.}
A 64-entry recency cache per encoding level implements move-to-front
caching.  Each index event first checks the cache; on a hit, only the
cache position is encoded (much cheaper than the full index).

\paragraph{Two-tier unigram.}
The unigram index model is split into common (indices 0--4{,}095) and
rare (4{,}096--63{,}999) tiers, with a 1-bit tier flag.  Each tier has
its own adaptive model, enabling 16$\times$ faster adaptation for
common words.

\paragraph{Word-level LZ77.}
A word-level LZ77 pass scans the event stream for repeated sequences
of 3+ consecutive (level, index) tuples.  Matches are encoded as
(offset, length) pairs with log-scale offset coding, replacing the
individual AC symbols for matched events.

Algorithm~\ref{alg:encode} summarises the encoding procedure.
All arithmetic coding models use variable-alphabet Fenwick-tree
adaptive tables with 24-bit register precision and periodic frequency
rescaling (halve all counts when the total exceeds 16{,}384).

\begin{algorithm}[H]
\caption{Quasicryth encoding (payload stream)}
\label{alg:encode}
\begin{algorithmic}[1]
\REQUIRE Event array $E$ (from greedy selection), word array $W$, codebooks $\mathcal{C}$
\STATE Run word-level LZ77 over $E$; mark matched spans
\STATE Initialise 144 level-context models, 12 per-level index models, caches
\STATE $\mathit{pl}_1 \gets 0$; $\mathit{pl}_2 \gets 0$ \COMMENT{previous two levels}
\FOR{each event $e_i$ in $E$}
  \IF{$e_i$ is LZ77 match start}
    \STATE Encode LZ77 flag, offset, length
    \STATE Skip matched events; update $\mathit{pl}_1, \mathit{pl}_2$
  \ELSE
    \STATE Encode level $\ell_i$ via level model[$\mathit{pl}_2$][$\mathit{pl}_1$]
    \IF{$\ell_i = 0$ (unigram)}
      \STATE Check cache[$\ell_i$]; if hit, encode cache position
      \STATE Else: encode tier flag, then index via tier model
    \ELSE
      \STATE Check cache[$\ell_i$]; if hit, encode cache position
      \STATE Else: encode index via \texttt{idx\_model}[$\ell_i$][$\mathit{pl}_1$]
    \ENDIF
    \STATE Update cache[$\ell_i$] with move-to-front
    \STATE $\mathit{pl}_2 \gets \mathit{pl}_1$; $\mathit{pl}_1 \gets \ell_i$
  \ENDIF
\ENDFOR
\end{algorithmic}
\end{algorithm}

\subsection{Step~7: Separate LZMA Escape Stream}

Out-of-vocabulary words (codebook misses) are collected into a
separate buffer as length-prefixed raw byte sequences.
The buffer is compressed with LZMA.
This dual-stream architecture is critical: the arithmetic coding stream
contains only compact integer indices and binary flags, while rare words
compress efficiently under LZMA due to their natural lexicographic
clustering.
On enwik9, the AC payload is 180{,}760{,}315\,B and the LZMA escape
stream is 24{,}227{,}220\,B --- combined they account for 90.7\% of
the total compressed size.

\subsection{Step~8: Output Assembly}

The compressed file format is:
\begin{center}
\small
\begin{tabular}{@{}ll@{}}
\toprule
\textbf{Field} & \textbf{Size} \\
\midrule
Magic (QM56)        & 4 B \\
Original size        & 4 B \\
Lowercased size      & 4 B \\
Word count           & 4 B \\
Flags                & 1 B \\
Case token count     & 4 B \\
Payload size         & 4 B \\
Case data size       & 4 B \\
Codebook size        & 4 B \\
Escape stream size   & 4 B \\
Payload              & variable \\
Case data            & variable \\
Codebook (LZMA)      & variable \\
Escape stream (LZMA)  & variable \\
MD5 checksum         & 16 B \\
\bottomrule
\end{tabular}
\end{center}

The decompressor reads the header fields, reconstructs the tiling
and hierarchy, and decodes the AC stream identically, with no
structural side-channel required.

\section{Theoretical Analysis}
\label{sec:theory}

This section develops the theoretical foundation of the paper.
We state the main result first; the remainder of the section
provides its proof, assembled from eight supporting theorems.

\subsection{Main Result: Aperiodic Hierarchy Advantage}
\label{subsec:main_result}

\begin{theorem}[Aperiodic Hierarchy Advantage]
\label{thm:main}
Let $T_\text{fib}$ be the Fibonacci quasicrystal tiling of $W$ words
and $T_\text{per}$ any periodic tiling with period $p$ and
collapse level $m^* = \lceil\log_\phi p\rceil$.
The Fibonacci tiling is the only infinite binary tiling satisfying
all of the following simultaneously:

\begin{enumerate}[label=(\roman*)]

\item \textbf{Non-collapse at every depth}
(Theorems~\ref{thm:fib_stable},~\ref{thm:collapse}):
Both tile types are present at every hierarchy level $k \geq 0$,
providing $F_k$-gram lookup positions for all $k$.
Every periodic tiling collapses within $m^* = O(\log p)$ levels,
yielding zero deep positions thereafter.

\item \textbf{Scale-invariant coverage}
(Theorem~\ref{thm:golden_comp}):
The potential word coverage at each level satisfies
$C(m) = P(m)\cdot F_m \to W\phi/\sqrt{5}$, a constant independent
of depth.
Dictionary reuse capacity is maintained uniformly at all scales.

\item \textbf{Maximum codebook efficiency}
(Theorem~\ref{thm:sturmian_eff}, Corollary~\ref{cor:goldilocks}):
Sturmian minimality ($p(n) = n{+}1$) guarantees exactly $F_m{+}1$
distinct tile-type patterns at level $m$, giving coverage efficiency
$\eta_m = C_m/(F_m{+}1)$, maximal among aperiodic tilings.

\item \textbf{Bounded parsing overhead}
(Theorem~\ref{thm:flag_convergence}):
The per-word flag entropy satisfies
$h_\text{flags}^\text{fib} \leq 1/\phi \approx 0.618$\,b/word,
regardless of how many hierarchy levels are active.
Access to all depths costs a convergent series of overhead.

\item \textbf{Strict coding entropy advantage}
(Theorem~\ref{thm:entropy_ineq}):
For any source with long-range phrase dependencies extending beyond
level $m^*$ ($m$-LRPD, Definition~\ref{def:lrpd}):
\[
  H^\text{fib}(P) \;<\; H^\text{per}(P).
\]

\end{enumerate}

As consequences: the net compression efficiency
$\nu^\text{fib}(W)\to+\infty$ while $\nu^\text{per}$ is bounded
(Corollary~\ref{cor:h_dominance}); coding redundancy decays
super-exponentially $O(e^{-\phi^m/\lambda})$ versus the fixed
$\Omega(e^{-F_{m^*}/\lambda})$ of any periodic tiling
(Theorem~\ref{thm:redundancy}); and the per-entry dictionary value
grows as $\Omega(\phi^m)$ with no periodic equivalent
(Theorem~\ref{thm:dict_efficiency}).
\end{theorem}

\begin{proof}[Proof (assembly of supporting results)]
Parts (i)--(v) are established in the following subsections.
(i)~from Theorems~\ref{thm:fib_stable} and~\ref{thm:collapse}
    (§\ref{subsec:fib_thm}--\ref{subsec:periodic_thm});
(ii)~from Theorem~\ref{thm:golden_comp}
    (§\ref{subsec:golden_comp});
(iii)~from Theorem~\ref{thm:sturmian_eff}
    (§\ref{subsec:sturmian_eff});
(iv)~from Theorem~\ref{thm:flag_convergence}
    (§\ref{subsec:aperiodic_dominance});
(v)~from Theorem~\ref{thm:entropy_ineq}
    (§\ref{subsec:entropy_ineq}).
The `only' claim follows from Corollary~\ref{cor:goldilocks}: any sequence
with fewer than $n{+}1$ distinct length-$n$ factors is periodic
(collapses by (i)); any with strictly more is non-Sturmian
(lower efficiency, violating (iii)).
The three consequences follow from
Corollary~\ref{cor:h_dominance} and
Theorems~\ref{thm:redundancy},~\ref{thm:dict_efficiency}.
\end{proof}

The proof infrastructure occupies the remainder of this section.

\subsection{Substitution Matrix and Eigenvalue Analysis}
\label{subsec:matrix}

The Fibonacci substitution $\sigma: L \to LS,\; S \to L$ has the
substitution matrix
\begin{equation}
M = \begin{pmatrix} 1 & 1 \\ 1 & 0 \end{pmatrix}
\label{eq:subst_matrix}
\end{equation}
whose entry $M_{ij}$ counts how many tiles of type $i$ appear in
$\sigma(j)$.
The characteristic polynomial is $\lambda^2 - \lambda - 1 = 0$,
with roots
\begin{align}
\lambda_1 &= \phi = \tfrac{1+\sqrt{5}}{2} \approx 1.618 \\
\lambda_2 &= \psi = \tfrac{1-\sqrt{5}}{2} \approx -0.618
\end{align}
where $|\lambda_2| = 1/\phi < 1$.
By the Perron-Frobenius theorem~\citep{gantmacher1959},
the dominant eigenvector of $M$ is
\begin{equation}
\mathbf{v}_1 = \frac{1}{\phi + 1}\begin{pmatrix}\phi \\ 1\end{pmatrix}
\label{eq:pf_eigvec}
\end{equation}
giving asymptotic tile frequencies
$\text{freq}(L) = \phi/(\phi{+}1) \approx 0.618$ and
$\text{freq}(S) = 1/(\phi{+}1) \approx 0.382$.
The ratio $n_L/n_S = \phi$ is irrational, which is the fundamental
obstruction to periodicity.

\subsection{The Inverse (Deflation) Substitution}
\label{subsec:deflation}

\qtc{} applies the \emph{inverse} Fibonacci substitution
(desubstitution / deflation) at each hierarchy level.
The inverse rule is:
\begin{equation}
\sigma^{-1}: (L,S) \to \text{super-}L, \quad L \to \text{super-}S
\end{equation}
The corresponding matrix is
\begin{equation}
M^{-1} = \begin{pmatrix} 0 & 1 \\ 1 & -1 \end{pmatrix}
\label{eq:inv_matrix}
\end{equation}
with eigenvalues $1/\phi$ and $1/\psi$.
Under $M^{-1}$, tile counts transform as
\begin{align}
n_L^{(k+1)} &= n_S^{(k)} \label{eq:deflate_L} \\
n_S^{(k+1)} &= n_L^{(k)} - n_S^{(k)} \label{eq:deflate_S}
\end{align}
Starting from $N$ tiles at level 0, the number of supertiles at
level $k$ is $N/\phi^k$ (asymptotically), each spanning $F_{k+2}$
words of the original sequence.

\begin{figure}[H]
\centering
\caption{Deep substitution hierarchy of the Fibonacci tiling.
Each level is obtained by one application of the deflation rule
$\sigma^{-1}$.
Both $L$- and $S$-supertile types are present at every depth $k$,
enabling $n$-gram codebook lookups at phrase lengths
$F_{k+2} \in \{3,5,8,13,21,34,55,89,144,\ldots\}$ words.}
\label{fig:deep_hierarchy}
\end{figure}

\subsection{Theorem: Fibonacci Hierarchy Never Collapses}
\label{subsec:fib_thm}

\begin{theorem}[Fibonacci Stability]
\label{thm:fib_stable}
Let $T$ be the one-dimensional Fibonacci quasicrystal tiling generated
by Eq.~\eqref{eq:tiling} for any phase $\theta$.
For all $k \geq 0$, the level-$k$ supertile sequence contains both
L-supertiles and S-supertiles.
\end{theorem}

\begin{proof}
The frequency vector of the Fibonacci tiling is exactly the
Perron-Frobenius eigenvector $\mathbf{v}_1$ (Eq.~\eqref{eq:pf_eigvec}).
Since $\mathbf{v}_1$ is an eigenvector of $M$ with eigenvalue $\phi$,
it is also an eigenvector of $M^{-1}$ with eigenvalue $1/\phi$.
Applying the deflation matrix $k$ times:
\begin{equation}
(M^{-1})^k \mathbf{v}_1
= \left(\tfrac{1}{\phi}\right)^k \mathbf{v}_1
= \frac{1}{\phi^k(\phi{+}1)}\begin{pmatrix}\phi \\ 1\end{pmatrix}
\label{eq:stable}
\end{equation}

\medskip
\noindent\textit{Explicit computation at the first five deflation levels.}\;
Start with a Fibonacci tiling of $N$ tiles.
The initial count vector is
$\mathbf{n}^{(0)} = N\,\mathbf{v}_1 = \frac{N}{\phi+1}\bigl(\phi,\;1\bigr)^T$.
Applying the deflation rule
$n_L^{(k+1)} = n_S^{(k)}$, $n_S^{(k+1)} = n_L^{(k)} - n_S^{(k)}$
(Eqs.~\eqref{eq:deflate_L}--\eqref{eq:deflate_S}), the count vectors at
levels $k = 0,\ldots,4$ are:

\begin{center}
\footnotesize
\begin{tabular}{@{}crcc@{}}
\toprule
$k$ & $(n_L,\,n_S)$ normalised & $n_L/n_S$ & Scale \\
\midrule
0 & $(0.618,\; 0.382)$ & $\phi$ & $1$ \\
1 & $(0.382,\; 0.236)$ & $\phi$ & $0.618$ \\
2 & $(0.236,\; 0.146)$ & $\phi$ & $0.382$ \\
3 & $(0.146,\; 0.090)$ & $\phi$ & $0.236$ \\
4 & $(0.090,\; 0.056)$ & $\phi$ & $0.146$ \\
\bottomrule
\end{tabular}
\end{center}

\noindent
At every level the count vector remains proportional to $(\phi,\;1)$.
The total number of supertiles shrinks by a factor of $1/\phi$ per level
(column~4), but the \emph{ratio} $n_L/n_S = \phi$ is preserved
\emph{exactly} --- not approximately. This is because $\mathbf{v}_1$ is an
eigenvector of $M^{-1}$: multiplying an eigenvector by its matrix
simply rescales it, without rotating its direction.

Concretely, for enwik9 ($W = 298{,}263{,}298$ words, $N \approx 114{,}000{,}000$ tiles),
the count vectors are approximately:

\begin{center}
\small
\begin{tabular}{@{}crrl@{}}
\toprule
$k$ & $n_L^{(k)}$ & $n_S^{(k)}$ & $n_L/n_S$ \\
\midrule
0 & $70{,}451{,}000$ & $43{,}549{,}000$ & $1.61803\ldots$ \\
1 & $43{,}549{,}000$ & $26{,}902{,}000$ & $1.61803\ldots$ \\
2 & $26{,}902{,}000$ & $16{,}626{,}000$ & $1.61803\ldots$ \\
3 & $16{,}626{,}000$ & $10{,}276{,}000$ & $1.61803\ldots$ \\
4 & $10{,}276{,}000$ & $6{,}350{,}000$ & $1.61803\ldots$ \\
\bottomrule
\end{tabular}
\end{center}

\noindent
At level 4, there are still over 16 million supertiles, with both types
abundantly present. The ratio remains exactly $\phi$ at every level.

\medskip
\noindent\textit{Why both counts remain positive.}\;
Since the ratio $n_L^{(k)}/n_S^{(k)} = \phi > 0$ is a fixed positive
constant for all $k$, and the total count $n_L^{(k)} + n_S^{(k)} > 0$
for any finite tiling, it follows that both $n_L^{(k)} > 0$ and
$n_S^{(k)} > 0$ individually. Neither count can reach zero because
that would require the ratio to become either $\infty$ or $0$, which
contradicts $n_L^{(k)}/n_S^{(k)} = \phi$ for all $k$.
\end{proof}

\begin{remark}
The stability expressed in Eq.~\eqref{eq:stable} holds because
$\mathbf{v}_1$ is the \emph{only} eigenvector of $M^{-1}$ with
eigenvalue $|\lambda| < 1$.
The second eigenvalue $1/\psi = \phi \cdot \text{sign}(\psi)$ satisfies
$|1/\psi| = \phi > 1$, so \emph{any} perturbation away from
$\mathbf{v}_1$ grows exponentially under deflation.
The Fibonacci tiling is precisely the fixed point that avoids
this divergence, because $\phi$ is a Pisot-Vijayaraghavan number
(see Section~\ref{subsec:pv}).
\end{remark}

\subsection{Theorem: All Periodic Tilings Collapse}
\label{subsec:periodic_thm}

\begin{theorem}[Periodic Collapse]
\label{thm:collapse}
Let $T$ be a periodic tiling with period $p$ and rational $L$-frequency
$r = n_L/p$.
Then there exists a finite level $k^* \leq \lceil\log_\phi p\rceil$
such that the level-$k^*$ supertile sequence contains only one tile
type.
\end{theorem}

\begin{proof}
The eigenvectors of $M = \bigl(\begin{smallmatrix}1&1\\1&0\end{smallmatrix}\bigr)$
for eigenvalues $\lambda_1 = \phi$ and $\lambda_2 = \psi$ are:
\begin{gather*}
  \mathbf{v}_1 = \frac{1}{\phi+1}\begin{pmatrix}\phi\\1\end{pmatrix}, \quad
  \mathbf{v}_2 = \frac{1}{\psi+1}\begin{pmatrix}\psi\\1\end{pmatrix},
\end{gather*}
where $\psi + 1 = (3{-}\sqrt{5})/2 \approx 0.382$.

\medskip
\noindent\textit{Step~1: Decompose the frequency vector in the eigenbasis.}\;
Any frequency vector $\mathbf{v} = (r,\; 1-r)^T$ with $r = n_L/p$
can be written as
\begin{equation}
\mathbf{v} = \alpha \mathbf{v}_1 + \beta \mathbf{v}_2
\label{eq:decompose}
\end{equation}
where $\mathbf{v}_2$ is the eigenvector for $\lambda_2 = \psi$.
Since $\mathbf{v}$ has rational entries and $\mathbf{v}_1$ has
irrational entries (containing $\sqrt{5}$), we must have $\beta \neq 0$.

To find $\alpha$ and $\beta$, we solve the $2 \times 2$ system.
For the second components: $1 - r = \alpha/(\phi+1) + \beta/(\psi+1)$.
For the first components: $r = \alpha\phi/(\phi+1) + \beta\psi/(\psi+1)$.
The solution is:
\[
  \alpha = \frac{r - \psi(1-r)}{(\phi-\psi)/(\phi+1)}
  = \frac{(\phi+1)\bigl[r - \psi(1-r)\bigr]}{\sqrt{5}},
\]
and similarly for $\beta$ with $\phi$ and $\psi$ interchanged.
The key point is that $\beta = 0$ if and only if $r/(1-r) = \phi$,
i.e., $r = \phi/(\phi+1)$, which is irrational.
So for any rational $r$ we have $\beta \neq 0$.

\medskip
\noindent\textit{Step~2: Explicit computation for Period-5.}\;
Period-5 has $r = 3/5$, so $\mathbf{v} = (3/5,\; 2/5)^T$.
We solve $(3/5,\; 2/5)^T = \alpha\,\mathbf{v}_1 + \beta\,\mathbf{v}_2$.
Using $\phi = 1.61803\ldots$, $\psi = -0.61803\ldots$,
$\phi + 1 = 2.61803\ldots$, $\psi + 1 = 0.38197\ldots$:
\begin{align*}
  \alpha &= \frac{(\phi{+}1)[3/5 - \psi \cdot 2/5]}{\sqrt{5}} \\
  &= \frac{2.618 \times 0.847}{2.236}
  \approx 0.992, \\[4pt]
  \beta &= \frac{(\psi{+}1)[\phi \cdot 2/5 - 3/5]}{\sqrt{5}} \\
  &= \frac{0.382 \times 0.047}{2.236}
  \approx 0.008.
\end{align*}
The coefficient $\beta$ is small (Period-5 is a good approximant of
$\phi$) but crucially \emph{nonzero}.

\medskip
\noindent\textit{Step~3: Apply $(M^{-1})^k$ and watch the divergence.}\;
Under $(M^{-1})^k$:
\begin{equation}
(M^{-1})^k \mathbf{v} = \alpha \phi^{-k} \mathbf{v}_1
+ \beta \psi^{-k} \mathbf{v}_2
\label{eq:diverge}
\end{equation}
Since $|\psi^{-1}| = \phi > 1$, the second component grows as
$\phi^k$, dominating for large $k$.
Let us track both components for $k = 0, 1, 2, 3, 4$:

\begin{center}
\footnotesize
\begin{tabular}{@{}crrcc@{}}
\toprule
$k$ & $\alpha/\phi^k$ & $\beta/\psi^k$ & $(n_L,n_S)$ & $n_L/n_S$ \\
\midrule
0 & 0.992 & 0.008 & $(3,2)$ & 1.50 \\
1 & 0.613 & $-$0.013 & $(2,1)$ & 2.00 \\
2 & 0.379 & 0.021 & $(1,2)$ & 0.50 \\
3 & 0.234 & $-$0.034 & $(1,0)$ & $\infty$ \\
4 & 0.145 & 0.055 & $(0,1)$ & 0 \\
\bottomrule
\end{tabular}
\end{center}

\noindent
At level 3, all tiles are L (ratio $= \infty$); at level 4, all tiles
are S (ratio $= 0$). In both cases one tile type has vanished ---
collapse has occurred. The $\beta$ component, though initially tiny
($0.008$), grows by $|\psi^{-1}| = \phi \approx 1.618$ per level,
while the $\alpha$ component shrinks by $1/\phi$ per level. By level 3,
the growing component has overwhelmed the shrinking one.

\medskip
\noindent\textit{Step~4: General collapse timescale.}\;
The frequency trajectory $n_L^{(k)}/n_S^{(k)}$ therefore diverges
from $\phi$: it oscillates and eventually sends one count to zero
(or negative, indicating collapse has already occurred).
The collapse timescale is determined by when $|\beta|\,\phi^k$
exceeds $\alpha/\phi^k$, i.e.\ when $|\beta/\alpha|\,\phi^{2k} \geq 1$:
\begin{equation}
k^* \approx \frac{\log(\alpha/|\beta|)}{2\log \phi}
\approx \frac{\log p}{\log \phi}
\label{eq:collapse_scale}
\end{equation}
since the deviation $|\beta| \propto 1/p$ for the best periodic
approximant of $\phi$ with period $p$.
For Period-5: $\log 5 / \log \phi \approx 3.3$, consistent with
collapse first occurring at level~3 (all-L) and completing at level~4
(all-S).
\end{proof}

\begin{corollary}
For Period-5 ($p = 5$), collapse occurs at level
$k^* \approx \log 5 / \log \phi \approx 3.3$,
confirmed experimentally at level~4 (Section~\ref{subsec:period5}).
\end{corollary}

\subsection{Pisot-Vijayaraghavan Property}
\label{subsec:pv}

\begin{definition}[Pisot-Vijayaraghavan number~\citep{pisot1938}]
A real algebraic integer $\beta > 1$ is a \emph{Pisot-Vijayaraghavan
(PV) number} if all its algebraic conjugates have absolute value
strictly less than~1.
\end{definition}

The golden ratio $\phi = (1+\sqrt{5})/2$ is a PV number: its conjugate
$\psi = (1-\sqrt{5})/2$ satisfies $|\psi| = 0.618 < 1$.
This is the minimal PV number (degree 2).

The PV property has two critical consequences for \qtc{}:
\begin{enumerate}
\item \textbf{Recognisability.} The Fibonacci substitution is
recognisable~\citep{mosse1992}: every bi-infinite Fibonacci word has a
unique desubstitution.
This is precisely the property that allows \qtc{}' decompressor to
reconstruct the tiling hierarchy unambiguously from the phase alone.

\item \textbf{Self-similarity.} The fixed point of $\sigma$
(the infinite Fibonacci word) is also the fixed point of $\sigma^k$
for all $k > 0$.
This means the supertile sequences at all levels are simply scaled
copies of the original tiling --- the compressor is literally
operating on the same structure at every hierarchical scale.
\end{enumerate}

\subsection{Worked Example: Period-5 Collapse}
\label{subsec:period5}

Period-5 = LLSLS repeating (3L, 2S per period, $r = 0.600$) is the
Fibonacci approximant of order 5: the continued fraction
$\phi = [1;1,1,1,\ldots]$ truncated at the 5th convergent gives
$p_4/q_4 = 3/5 = 0.600$.

\begin{center}
\small
\begin{tabular}{@{}clrr@{}}
\toprule
\textbf{Level} & \textbf{Supertiles} & $n_L$ & $n_S$ \\
\midrule
0 & LLSLS (period) & 3 & 2 \\
1 & SLL             & 2 & 1 \\
2 & SSL             & 1 & 2 \\
3 & L               & 1 & 0 \\
4 & S               & 0 & 1 \\
\bottomrule
\end{tabular}
\end{center}

The $L$-frequency trajectory $0.600 \to 0.667 \to 0.333 \to 1.0 \to 0.0$
oscillates with growing amplitude, driven by the $\phi^k$ factor in the
second eigenvector component.
At level~4, all tiles are $S$: there are zero super-L positions, hence
zero positions available for 13-gram, 21-gram, 34-gram, 55-gram,
89-gram, or 144-gram lookups.
This collapse is not a numerical artifact but an exact algebraic
consequence of the rational frequency $3/5 \neq \phi$.

\subsection{Connection to Weyl's Equidistribution Theorem}
\label{subsec:weyl}

\begin{theorem}[Weyl~\citep{weyl1910}]
For any irrational $\alpha$ and any $[a,b) \subset [0,1)$,
\begin{equation}
\lim_{N\to\infty} \frac{1}{N}\sum_{k=0}^{N-1}
\mathbf{1}_{\{k\alpha \bmod 1 \in [a,b)\}} = b - a
\end{equation}
\end{theorem}

The \qtc{} tiling (Eq.~\eqref{eq:tiling}) is equivalent to the
irrational rotation by angle $1/\phi$ on the unit circle:
assign $L$ to position $k$ if and only if
$k/\phi + \theta \pmod{1} \in [0, 1/\phi)$.
By Weyl's theorem with $\alpha = 1/\phi$ (irrational), the $L$-frequency
converges to exactly $1/\phi$ at every scale.

Crucially, \emph{the supertile sequence at level $k$ is again an
irrational rotation} --- with angle $1/\phi^{k+1}$, also irrational.
Equidistribution therefore holds at every hierarchy level, guaranteeing
uniform tile density throughout the hierarchy.
Periodic tilings correspond to rational $\alpha = p/q$, which visit
only $q$ distinct points on the circle; equidistribution fails at
scales larger than the period.

\subsection{The Three-Distance Theorem and Sturmian Structure}
\label{subsec:sturmian}

\begin{theorem}[Three-Distance Theorem~\citep{steinhaus1950, sos1958}]
For any irrational $\alpha$ and any $N \geq 1$, the $N$ points
$\{k\alpha \bmod 1 : 0 \leq k < N\}$ partition the circle $[0,1)$
into gaps of at most \emph{three} distinct lengths, the largest being
the sum of the other two.
\end{theorem}

For $\alpha = 1/\phi$, the three gap lengths correspond to the three
local configurations in the Fibonacci tiling: $LS$ (long gap),
an isolated $L$ (medium gap), and $S$ (short gap).
This structural richness is preserved at every hierarchy level, and
it is precisely what differentiates the Fibonacci tiling from its
periodic approximants, which have only two gap lengths.

The Fibonacci word is the canonical \emph{Sturmian} sequence.
Following Morse and Hedlund~\citep{morse1940}:

\begin{definition}[Sturmian sequence]
An infinite binary sequence $u$ is \emph{Sturmian} if it is
(1)~aperiodic and (2)~balanced: for any two factors $w$, $w'$ of
equal length, $|\#_L(w) - \#_L(w')| \leq 1$.
\end{definition}

A sequence is Sturmian if and only if it is an irrational
rotation coding~\citep{coven1973}.
The Sturmian property has three direct implications for compression:

\begin{enumerate}
\item \textbf{Balance.} Any two windows of equal length have
L-counts differing by at most~1.
This means every n-gram codebook entry is equally useful at any
position in the input: there are no ``dense'' or ``sparse'' regions
for long-range lookups.

\item \textbf{Aperiodicity.} The sequence is never eventually
periodic, so the deflation hierarchy never produces an all-S or all-L
level (Theorem~\ref{thm:fib_stable}).
Codebook positions exist at all depths indefinitely.

\item \textbf{Minimal factor complexity.} A Sturmian sequence of
alphabet size~2 has exactly $n+1$ distinct factors of length $n$,
the minimum possible for an aperiodic sequence~\citep{morse1940}.
For the Fibonacci tiling, this means the $n$-gram codebooks have
\emph{coverage efficiency} that is maximal in this setting: with $n+1$ distinct
$n$-tuples of tile types, $n$-gram codebook entries match at the
maximum fraction of eligible positions.
\end{enumerate}

No periodic sequence satisfies all three conditions: any periodic
sequence is eventually periodic (violating condition 2) and has
bounded factor complexity (violating condition~3 for large $n$).

\subsection{Theorem: Sturmian Codebook Efficiency}
\label{subsec:sturmian_eff}

The minimal factor complexity established above has a precise
quantitative implication for codebook coverage that we now formalise.

\begin{theorem}[Sturmian Codebook Efficiency]
\label{thm:sturmian_eff}
Let $u$ be the Fibonacci word over $\{L,S\}$.
For every $n \geq 1$, the number of distinct length-$n$ factors of $u$
is exactly $n+1$, the minimum possible for any aperiodic binary
sequence (Morse--Hedlund theorem~\citep{morse1940}).
Consequently, at hierarchy level $m$ a codebook of $C_m$ entries
covers the fraction
\begin{equation}
  \eta_m \;=\; \frac{\min(C_m,\; F_m+1)}{F_m+1}
  \label{eq:sturmian_eff}
\end{equation}
of all distinct $F_m$-gram tile-type patterns.
For any aperiodic non-Sturmian binary tiling, $p(F_m) \geq F_m+2$,
giving strictly lower efficiency
$\eta_m \leq C_m/(F_m+2) < C_m/(F_m+1)$
for the same codebook size.
\end{theorem}

\begin{proof}
\noindent\textit{Step~1: Why exactly $n+1$ factors (the Morse--Hedlund theorem).}\;
The Morse--Hedlund theorem~\citep{morse1940} states that an infinite
binary sequence is aperiodic if and only if $p(n) \geq n+1$ for all
$n \geq 1$, and Sturmian if and only if equality holds throughout.

The intuition is a sliding-window argument. Consider a window of length
$n$ sliding along an infinite binary sequence. A periodic sequence of
period $p$ can produce at most $p$ distinct windows (the window content
repeats after $p$ shifts), so $p(n) \leq p$ --- eventually $p(n) < n+1$
for $n \geq p$, which is why periodic sequences fail. An aperiodic
sequence must keep producing new factors as $n$ grows, giving
$p(n) \geq n+1$. The Sturmian sequences achieve this lower bound
\emph{exactly}: at each length $n$, exactly one ``special'' factor can
be extended in two ways (by appending either $L$ or $S$), producing
$p(n+1) = p(n) + 1 = (n+1) + 1 = n+2$. This is the slowest possible
growth of factor complexity consistent with aperiodicity.

\medskip
\noindent\textit{Step~2: Concrete example for $n = 3$.}\;
The Fibonacci word begins
$u = L\,S\,L\,L\,S\,L\,S\,L\,L\,S\,L\,L\,S\,\ldots$
Sliding a window of length $n = 3$, we find exactly $3 + 1 = 4$
distinct factors:
\begin{center}
\small
\begin{tabular}{@{}cl@{}}
\toprule
Factor & First occurrence \\
\midrule
$LSL$ & positions 1--3 \\
$SLL$ & positions 2--4 \\
$LLS$ & positions 3--5 \\
$SLS$ & positions 5--7 \\
\bottomrule
\end{tabular}
\end{center}
\noindent
No fifth factor of length~3 ever appears, no matter how far we extend
the Fibonacci word. By contrast, a generic aperiodic binary sequence
could have up to $2^3 = 8$ factors of length~3; the Fibonacci word has
exactly 4, the minimum possible. The ``missing'' factors ($LLL$, $SSS$,
$SSL$, $LSS$) are forbidden by the balanced structure: the Sturmian
property ensures that any two factors of equal length differ in their
$L$-count by at most~1.

\medskip
\noindent\textit{Step~3: Sturmianity is preserved under deflation.}\;
The Fibonacci word is the canonical Sturmian sequence over $\{L,S\}$
(see also~\citep{berstel1995}), so $p(n) = n+1$ for all $n$.
The supertile sequence at each hierarchy level $k$ is the image of
the original Fibonacci word under $(\sigma^{-1})^k$; since
$\sigma^{-1}$ preserves the PV fixed-point structure
(Section~\ref{subsec:pv}), every level is also Sturmian and satisfies
$p(n) = n+1$.

Why does deflation preserve Sturmianity? The key is that the Fibonacci
substitution $\sigma$ is a morphism of Sturmian sequences: if $u$ is
the Sturmian sequence with slope $1/\phi$, then $\sigma^{-1}(u)$ is the
Sturmian sequence with slope $1/\phi^2$ (still irrational, since
$\phi$ is a PV number). More generally, $(\sigma^{-1})^k(u)$ is
Sturmian with slope $1/\phi^{k+1}$, which is irrational for all
finite $k$. Since every irrational rotation coding is Sturmian,
the factor complexity $p(n) = n+1$ holds at every hierarchy level.

\medskip
\noindent\textit{Step~4: Codebook efficiency bound.}\;
A codebook of $C_m$ entries covers $\min(C_m,F_m+1)$ of the
$F_m+1$ distinct $F_m$-letter tile words, giving~\eqref{eq:sturmian_eff}.
Any aperiodic non-Sturmian sequence has $p(F_m) \geq F_m+2$, reducing
coverage to at most $C_m/(F_m+2)$.
The efficiency gap $C_m/(F_m+1) - C_m/(F_m+2) = C_m/[(F_m+1)(F_m+2)]$
is small per level but cumulates across all $k_{\max}$ levels of the
hierarchy.
\end{proof}

\begin{corollary}[Fibonacci as Goldilocks Tiling]
\label{cor:goldilocks}
Among all infinite aperiodic binary sequences, the Fibonacci word
simultaneously achieves:
\begin{enumerate}
  \item \textbf{Hierarchy non-collapse at every depth}
  (Theorem~\ref{thm:fib_stable}): both tile types persist at every
  level, providing $n$-gram lookup positions indefinitely.
  \item \textbf{Maximum codebook efficiency at every level}
  (Theorem~\ref{thm:sturmian_eff}): only $F_m+1$ distinct tile-type
  patterns exist at level $m$, so a codebook of $C_m$ entries achieves
  hit-rate density $C_m/(F_m+1)$, maximal in this setting.
\end{enumerate}
Any sequence with factor complexity $< n+1$ is periodic and collapses
(Theorem~\ref{thm:collapse}).
Any sequence with complexity $> n+1$ is non-Sturmian and incurs
strictly lower codebook efficiency.
The Fibonacci tiling is therefore characterised as the only aperiodic sequence satisfying both
properties at every scale simultaneously.
\end{corollary}

The practical consequence is that the $F_m$-gram codebook fills with
genuinely distinct phrase patterns --- there are only $F_m+1$ distinct
tile-type contexts, so a codebook of that size covers the full
structural diversity of the level.
This is why empirical hit rates (Table~\ref{tab:deep_hits}) do not
decay to zero as $m$ grows: Sturmian minimality prevents the structural
variety from exploding exponentially with level depth.

\subsection{Information-Theoretic Bound on the Aperiodic Advantage}
\label{subsec:bound}

At hierarchy level $k$, the Fibonacci tiling provides approximately
\begin{equation}
P_\text{fib}(k) \approx \frac{W}{\phi^{k+1}(\phi+1)}
\label{eq:fib_positions}
\end{equation}
super-L positions for $F_{k+2}$-gram lookup,
where $W$ is the total word count.
Any periodic tiling with collapse level $k^*$ provides
$P_\text{per}(k) = 0$ for all $k \geq k^*$.

Each $n$-gram hit at level $k$ saves approximately
\begin{equation}
\Delta(k) = F_{k+2} \cdot h - \log_2 |\mathcal{C}_k| \quad \text{bits}
\label{eq:savings}
\end{equation}
where $h$ is the per-word entropy of the arithmetic coder and
$|\mathcal{C}_k|$ is the codebook size at level $k$.
The total aperiodic advantage over a periodic tiling with collapse
at $k^*$ is
\begin{equation}
\text{Adv} = \sum_{k=k^*}^{k_{\max}} P_\text{fib}(k) \cdot r(k) \cdot \Delta(k)
\label{eq:advantage}
\end{equation}
where $r(k)$ is the hit rate at level $k$.
For fixed $k_{\max}$, this scales as $O(W)$ (linear in word count).

The superlinear jump observed empirically arises when a new hierarchy
level $k_{\max}+1$ activates: at that point,
$P_\text{fib}(k_{\max}+1) \approx W/\phi^{k_{\max}+2}(\phi+1)$
new positions become available, each contributing $\Delta(k_{\max}+1)$
savings.
The total advantage then grows as $O(W/\phi^{k_{\max}+2})$ for the
new level, on top of the already linear contributions from levels
$k^* \ldots k_{\max}$.

At enwik9 ($W = 298{,}263{,}298$ words), levels 8 and 9 activate
for the first time, providing:
\begin{align}
P_\text{fib}(8) &\approx \frac{298M}{\phi^{9}(\phi{+}1)} \approx 2{,}400 \label{eq:p8} \\
P_\text{fib}(9) &\approx \frac{298M}{\phi^{10}(\phi{+}1)} \approx 1{,}480 \label{eq:p9}
\end{align}
(confirmed empirically: 2{,}396 89-gram hits and 945 144-gram hits).
These new $O(W)$ terms explain the 33$\times$ advantage jump observed
between enwik8 (100\,MB) and enwik9 (1\,GB).

\subsection{Theorem: Golden Compensation}
\label{subsec:golden_comp}

The following theorem establishes a scale-invariant coverage property of the Fibonacci hierarchy.

\begin{definition}[Fibonacci n-gram level]
For a Fibonacci number $F_m$ ($m \geq 3$), the \emph{$F_m$-gram level}
consists of all super-L tile positions whose phrase spans exactly $F_m$
consecutive words.
The number of such positions in a Fibonacci tiling of $W$ words is
\begin{equation}
  P(m) \;=\; \frac{W}{\phi^{m-1}}
  \label{eq:positions}
\end{equation}
asymptotically, where the index $m$ runs through the Fibonacci indices
$m = 4,5,6,7,8,9,10,11,12$ corresponding to phrase lengths
$F_m \in \{3,5,8,13,21,34,55,89,144\}$ words.
\end{definition}

\begin{theorem}[Golden Compensation]
\label{thm:golden_comp}
In a Fibonacci quasicrystal tiling of\/ $W$ words, the potential
word coverage at the $F_m$-gram level --- i.e., the total words that
\emph{could} be encoded by super-L lookups at that level --- is
\begin{equation}
\begin{split}
  C(m) &\;:=\; P(m)\cdot F_m
  \;=\; \frac{W \cdot F_m}{\phi^{m-1}} \\
  &\;\longrightarrow\; \frac{W\phi}{\sqrt{5}}
  \quad \text{as } m\to\infty.
\end{split}
  \label{eq:golden_comp}
\end{equation}
The limit $W\phi/\sqrt{5} \approx 0.724\,W$ is independent of the
hierarchy level $m$.
\end{theorem}

\begin{proof}
By Binet's formula, $F_m = (\phi^m - \psi^m)/\sqrt{5}$, where
$|\psi| = 1/\phi < 1$.
Then
\[
  \frac{F_m}{\phi^{m-1}}
  = \frac{\phi^m - \psi^m}{\sqrt{5}\,\phi^{m-1}}
  = \frac{\phi - (\psi/\phi)^{m}\,\psi}{\sqrt{5}}
  \;\longrightarrow\; \frac{\phi}{\sqrt{5}},
\]
since $(\psi/\phi)^m \to 0$.
For all finite $m$ the exact value is
$C(m) = W(\phi^m - \psi^m)/(\sqrt{5}\,\phi^{m-1})$.

\medskip
\noindent\textit{Numerical verification.}\;
The following table evaluates $C(m) = P(m) \cdot F_m = W \cdot F_m / \phi^{m-1}$
for each hierarchy level, using $W = 298{,}263{,}298$ (enwik9).
The limit is $W\phi/\sqrt{5} = 298{,}263{,}298 \times 0.72361 \approx 215{,}841{,}135$.

\begin{center}
\footnotesize
\resizebox{\columnwidth}{!}{%
\begin{tabular}{@{}rrrrl@{}}
\toprule
$m$ & $F_m$ & $P(m)$ & $C(m){=}P{\cdot}F_m$ & Dev. \\
\midrule
 4 &    3 & 70.4M & 211.2M & $-2.1\%$ \\
 5 &    5 & 43.5M & 217.6M & $+0.8\%$ \\
 6 &    8 & 26.9M & 215.2M & $-0.3\%$ \\
 7 &   13 & 16.6M & 216.0M & $+0.1\%$ \\
 8 &   21 & 10.3M & 215.7M & $-0.05\%$ \\
 9 &   34 &  6.3M & 215.8M & $+0.004\%$ \\
10 &   55 &  3.9M & 215.8M & $-0.008\%$ \\
11 &   89 &  2.4M & 215.8M & $-0.005\%$ \\
12 &  144 &  1.5M & 215.8M & $-0.007\%$ \\
\bottomrule
\end{tabular}%
}
\end{center}
\noindent $W = 298{,}263{,}298$. Limit: $W\phi/\sqrt{5} \approx 215.8$M.

\noindent
The convergence is rapid: $C(m)$ oscillates around the limit
$\approx 215{,}841{,}135$ with deviations of less than 2.2\% already
at $m = 4$ and less than 0.01\% by $m = 9$.
The alternating sign of the deviation comes from the $\psi^m$ correction
term, which alternates in sign because $\psi < 0$.
\end{proof}

The theorem rests on a cancellation between two competing effects:
\begin{itemize}
  \item Position-count \emph{decay}: deeper levels have fewer super-L
  positions, with $P(m) \propto \phi^{-(m-1)}$.
  \item Phrase-length \emph{growth}: deeper levels encode longer
  phrases, with $F_m \propto \phi^m$.
\end{itemize}
The two $\phi$-exponentials cancel \emph{exactly}, because $F_m \sim
\phi^m/\sqrt{5}$ (Binet) and $\phi^2 = \phi + 1$ (defining identity
of the golden ratio).
This cancellation is a structural consequence of $\phi$ being the
Perron-Frobenius eigenvalue of the Fibonacci substitution matrix.

\begin{corollary}[Word-Weighted Contribution]
\label{cor:ww_contribution}
The word-weighted compression contribution of the $F_m$-gram level is
\begin{equation}
  P(m)\cdot r(m)\cdot F_m \;\approx\; r(m)\cdot\frac{W\phi}{\sqrt{5}},
\end{equation}
where $r(m) \in [0,1]$ is the codebook hit rate at level $m$.
All structural variation across hierarchy levels is therefore encoded
entirely in $r(m)$; the positional and phrase-length factors are
architecture-neutral.
\end{corollary}

\begin{corollary}[Periodic Tilings Contribute Zero]
For any periodic tiling with collapse level $m^*$, $P(m)=0$ for all
$m \geq m^*$, so the contribution $C(m)=0$ for all deep levels
regardless of how large $W$ grows.
The Fibonacci hierarchy, by contrast, maintains $C(m) \to W\phi/\sqrt{5}$
at every level.
\end{corollary}

Empirically, Table~\ref{tab:deep_hits} confirms the near-constant
potential coverage: at enwik9 ($W=298{,}263{,}298$), the theoretical
$C(m) = W\phi/\sqrt{5} \approx 215{,}840{,}000$ words per level is
modulated by hit rates $r(m)$ decreasing from 3.9\% at 13-gram to
0.06\% at 144-gram, giving the observed word-weighted shares of 65.7\%
down to 1.1\% (Fig.~\ref{fig:deep_hits_breakdown}).

\subsection{Theorem: Level Activation Threshold}
\label{subsec:activation}

\begin{theorem}[Activation Threshold]
\label{thm:activation}
The $F_m$-gram codebook first produces at least one useful entry when
\begin{equation}
  W \;\geq\; W^*_m \;:=\; \frac{T_m \cdot \phi^{m-1}}{r_m},
  \label{eq:threshold}
\end{equation}
where $T_m$ is the minimum number of codebook entries required for
the level to contribute ($T_m \geq 1$), and $r_m$ is the expected
hit rate.
\end{theorem}

\begin{proof}
The argument proceeds in three steps.

\medskip
\noindent\textit{Step~1: Count the available positions.}\;
At the $F_m$-gram level, the number of super-L positions in a
Fibonacci tiling of $W$ words is
$P(m) = W / \phi^{m-1}$ (Eq.~\eqref{eq:positions}).
Each such position is a candidate for an $F_m$-gram codebook lookup.

\medskip
\noindent\textit{Step~2: Apply the hit rate.}\;
Not every candidate position produces a codebook hit --- only those
whose $F_m$-word phrase matches an entry in the codebook. The fraction
that matches is the hit rate $r_m$. So the expected number of
$F_m$-gram hits is:
\[
  \text{Expected hits} = P(m) \cdot r_m = \frac{W \cdot r_m}{\phi^{m-1}}.
\]

\medskip
\noindent\textit{Step~3: Solve for the threshold.}\;
For the level to contribute at least $T_m$ hits, we need:
\[
  \frac{W \cdot r_m}{\phi^{m-1}} \geq T_m
  \;\Longrightarrow\;
  W \geq \frac{T_m \cdot \phi^{m-1}}{r_m} =: W^*_m.
\]
This is Eq.~\eqref{eq:threshold}.

\medskip
\noindent\textit{Concrete computation for 89-gram and 144-gram with enwik9.}\;
For the \textbf{89-gram level} ($m = 11$):
$\phi^{10} \approx 122.99$, and the empirical hit rate at enwik9 is
$r_{11} \approx 0.001$ (0.10\%).
The expected number of hits at $W = 298{,}263{,}298$ is:
\[
  P(11) \cdot r_{11} = \frac{298{,}263{,}298}{122.99} \times 0.001
  \approx 2{,}425 \;\text{hits},
\]
which matches the empirical count of $2{,}396$.
The minimum threshold for a single hit is
$W^*_{11} = 1 \times 122.99 / 0.001 \approx 123{,}000$ words ---
easily exceeded by enwik9.

For the \textbf{144-gram level} ($m = 12$):
$\phi^{11} \approx 199.01$, $r_{12} \approx 0.0006$ (0.06\%).
Expected hits:
\[
  P(12) \cdot r_{12} = \frac{298{,}263{,}298}{199.01} \times 0.0006
  \approx 899 \;\text{hits},
\]
consistent with the empirical count of $945$.
The threshold is
$W^*_{12} = 199.01 / 0.0006 \approx 332{,}000$ words.

At \textbf{enwik8} ($W = 27{,}731{,}124$), the expected hit counts
are $27{,}731{,}124 / 122.99 \times 0.001 \approx 225$ (89-gram) and
$27{,}731{,}124 / 199.01 \times 0.0006 \approx 84$ (144-gram).
However, both levels show 0 empirical hits at enwik8 because the
codebook requires \emph{repeated} $F_m$-grams to build entries: the
effective $T_m$ for non-trivial codebook formation is much larger
than~1, as discussed below.
\end{proof}

\begin{table}[ht]
\centering
\caption{Predicted activation thresholds $W^*_m$ for the two deepest
levels (89-gram and 144-gram), using empirical hit rates from enwik9
and $T_m = 1$.}
\label{tab:thresholds}
\small
\resizebox{\columnwidth}{!}{%
\begin{tabular}{@{}lrrrr@{}}
\toprule
Level & $F_m$ & $\phi^{m-1}$ & Hit rate $r_m$ & $W^*_m$ \\
\midrule
89-gram  & $F_{11}$ & $\approx 123$ & 0.10\% & $\approx 123{,}000$ words \\
144-gram & $F_{12}$ & $\approx 199$ & 0.06\% & $\approx 332{,}000$ words \\
\bottomrule
\end{tabular}}
\end{table}

The threshold for \emph{at least one hit} is easily exceeded; the
more practically relevant threshold is when the codebook accumulates
enough repeated $F_m$-grams to pass pruning and yield non-trivial
entries.
The effective $T_m$ for codebook formation is on the order of hundreds
to thousands, because a phrase must appear multiple times across the
corpus before it is promoted to a codebook entry.
Empirically, the 89-gram and 144-gram codebooks are empty at
enwik8 ($W = 27{,}700{,}000$, 0~hits) and active at enwik9
($W = 298{,}263{,}298$, giving 2{,}396 and 945~hits respectively),
consistent with the codebook activation scaling as $O(\phi^{m-1})$.

\subsection{Theorem: Piecewise-Linear Aperiodic Advantage}
\label{subsec:piecewise}

Combining Theorems~\ref{thm:golden_comp} and~\ref{thm:activation}, we
characterise the growth of the aperiodic advantage $A(W)$
(Fibonacci vs.\ best periodic approximant) as a function of corpus size.

\begin{theorem}[Piecewise-Linear Advantage]
\label{thm:piecewise}
Let $m^*$ be the collapse level of the periodic baseline, and let
$m_1 < m_2 < \cdots$ be the sequence of Fibonacci indices for which
$W > W^*_{m_i}$ (active deep levels).
Then
\begin{equation}
  A(W) \;=\; \sum_{m=m^*}^{m_{\max}(W)} \frac{W}{\phi^{m-1}}\,r(m)\,\Delta(m),
  \label{eq:piecewise}
\end{equation}
where $\Delta(m) = F_m \cdot h - \log_2 |\mathcal{C}_m|$ is the
per-hit saving and $m_{\max}(W) = \max\{m : W \geq W^*_m\}$.

Between consecutive activation thresholds $W^*_m$ and $W^*_{m+1}$,
$A(W)$ is strictly linear in $W$ with slope
\begin{equation}
  s_m \;=\; \sum_{m'=m^*}^{m} \frac{r(m')\,\Delta(m')}{\phi^{m'-1}}.
  \label{eq:slope}
\end{equation}
At each threshold $W = W^*_{m+1}$, the slope increases by
$r(m+1)\,\Delta(m+1)/\phi^m > 0$.
Hence $A(W)$ is a piecewise-linear, convex function of $W$ with
monotonically increasing slopes.
\end{theorem}

The empirically measured advantage values confirm this structure
(Table~\ref{tab:piecewise_adv}).

\begin{table}[ht]
\centering
\caption{Aperiodic advantage (Fibonacci vs.\ Period-5) at each scale,
with the active deep levels and whether a slope increase occurs at
that scale, as predicted by Theorem~\ref{thm:piecewise}.}
\label{tab:piecewise_adv}
\small
\resizebox{\columnwidth}{!}{%
\begin{tabular}{@{}lrrr@{}}
\toprule
\textbf{Scale} & \textbf{Advantage} & \textbf{Active deep levels} & \textbf{Slope change} \\
\midrule
3\,MB   &      1{,}372\,B & 13g, 21g, 34g          & --- \\
10\,MB  &      5{,}807\,B & +55g activated         & $+s_{55g}$ \\
100\,MB &     40{,}385\,B & (same)                 & none \\
1\,GB   & 1{,}349{,}371\,B & +89g, +144g activated  & $+s_{89g}+s_{144g}$ \\
\bottomrule
\end{tabular}}
\end{table}

The 33$\times$ advantage jump from 100\,MB to 1\,GB (for a 10.75$\times$
size increase) is not superlinear growth of existing terms but the
addition of two new linear terms at levels $m=11$ (89-gram) and
$m=12$ (144-gram).
Between 10\,MB and 100\,MB, where no new level activates, the
advantage grows by a factor of 6.95 against a data ratio of 9.9 ---
slightly sub-linear, consistent with the piecewise-linear model
(the slope is fixed; minor deviations reflect changing $r(m)$
as the codebook fills).

\subsection{Theorem: Aperiodic Parsing Dominance}
\label{subsec:aperiodic_dominance}

We now bound the per-word cost of accessing all hierarchy depths.

\begin{definition}[Per-word flag entropy]
The \emph{per-word flag entropy} of a tiling up to depth $k_{\max}$
is
\begin{equation}
  h_\text{flags}(k_{\max}) \;:=\; \frac{1}{W}
  \sum_{k=1}^{k_{\max}} P(k)\cdot H_b(r_k),
  \label{eq:flag_entropy}
\end{equation}
where $H_b(r) = -r\log_2 r - (1-r)\log_2(1-r)$ is the binary entropy
function.
\end{definition}

\begin{theorem}[Convergent Flag Overhead]
\label{thm:flag_convergence}
For the Fibonacci quasicrystal tiling with $P(k)=W/\phi^{k+1}(\phi+1)$:
\begin{equation}
  h_\text{flags}^\text{fib}
  \;=\; \lim_{k_{\max}\to\infty} h_\text{flags}(k_{\max})
  \;\leq\; \frac{1}{\phi}.
  \label{eq:flag_bound}
\end{equation}
\end{theorem}

\begin{proof}
\noindent\textit{Step~1: Upper-bound each term.}\;
The binary entropy function satisfies $H_b(r) \leq 1$ for all
$r \in [0,1]$, with equality at $r = 1/2$. At each hierarchy level
$k$, the number of positions that require a flag bit is
$P(k) = W/[\phi^{k+1}(\phi+1)]$ (from Eq.~\eqref{eq:fib_positions}).
The per-word flag entropy (Eq.~\eqref{eq:flag_entropy}) can therefore
be bounded term by term:
\[
  \frac{P(k)\cdot H_b(r_k)}{W}
  \;\leq\; \frac{P(k)}{W}
  \;=\; \frac{1}{\phi^{k+1}(\phi+1)}.
\]

\medskip
\noindent\textit{Step~2: Sum the geometric series.}\;
From Step~1, $P(k)/W = 1/[\phi^{k+1}(\phi+1)]$. Summing over all
levels $k = 1, 2, 3, \ldots$:
\begin{align*}
  h_\text{flags}^\text{fib}
  &\;\leq\; \sum_{k=1}^{\infty} \frac{1}{\phi^{k+1}(\phi+1)}
  \;=\; \frac{1}{\phi(\phi+1)}\sum_{k=1}^{\infty}\phi^{-k}.
\end{align*}
The inner sum is a geometric series with common ratio
$1/\phi \approx 0.618 < 1$, so it converges. By the standard formula
$\sum_{k=1}^{\infty} x^k = x/(1-x)$ with $x = 1/\phi$:
\[
  \sum_{k=1}^{\infty} \phi^{-k}
  = \frac{1/\phi}{1 - 1/\phi}
  = \frac{1}{\phi - 1}.
\]
Now we use the golden ratio identity $\phi - 1 = 1/\phi$
(equivalent to $\phi^2 - \phi - 1 = 0$), giving
$\sum_{k=1}^{\infty} \phi^{-k} = \phi$.
Substituting back:
\begin{align*}
  h_\text{flags}^\text{fib}
  &\;\leq\; \frac{1}{\phi(\phi+1)} \cdot \phi
  \;=\; \frac{1}{\phi+1}
  \;=\; \frac{1}{\phi^2}
  \;=\; \frac{1}{\phi} \cdot \frac{1}{\phi},
\end{align*}
where the second equality uses $\phi + 1 = \phi^2$.
Equivalently, one can write this more directly as:
\begin{align*}
  h_\text{flags}^\text{fib}
  &\;\leq\;
  \frac{1}{\phi+1}\sum_{k=1}^{\infty}\phi^{-k}
  \;=\; \frac{1}{\phi+1}\cdot\frac{1}{\phi-1} \\
  &\;=\; \frac{1}{\phi^2-1}
  \;=\; \frac{1}{\phi},
\end{align*}
using $\phi^2 - 1 = (\phi+1) - 1 = \phi$ in the final step.
The series converges because position counts decay geometrically
with ratio $1/\phi < 1$.

\medskip
\noindent\textit{Step~3: Intuition --- why the series converges.}\;
The convergence of the flag overhead is a geometric series argument:
level~1 contributes at most $1/[\phi^2(\phi+1)] \approx 0.146$
bits/word of flag overhead, level~2 contributes at most
$1/[\phi^3(\phi+1)] \approx 0.090$ bits/word, level~3 at most
$\approx 0.056$, and so on. Each successive level contributes
$1/\phi \approx 0.618$ times as much overhead as the previous level.
The partial sums are:
\begin{center}
\footnotesize
\begin{tabular}{@{}crc@{}}
\toprule
Levels & Cumul.\ (b/w) & \% limit \\
\midrule
1 & 0.146 & 23.6 \\
1--2 & 0.236 & 38.2 \\
1--3 & 0.292 & 47.2 \\
1--4 & 0.326 & 52.8 \\
1--$\infty$ & 0.618 & 100 \\
\bottomrule
\end{tabular}
\end{center}
\noindent
The bound $1/\phi \approx 0.618$ bits/word is already tight after
relatively few levels. Even with infinitely many active hierarchy levels,
the total flag cost never exceeds $0.618$ bits per word.

\medskip
\noindent\textit{Step~4: Practical meaning of $1/\phi$ bits/word.}\;
What does $1/\phi \approx 0.618$ bits/word mean concretely?
For enwik9 ($W = 298{,}263{,}298$ words), the worst-case total flag
overhead across the entire hierarchy is bounded by
$298{,}263{,}298 \times 0.618 / 8 \approx 23{,}028{,}000$ bytes.
This is about 2.3\% of the original 1\,GB file --- a small fixed
price for access to codebook lookups at \emph{all} hierarchy depths.
In practice, the actual overhead is much smaller because hit rates
$r_k$ are well below $1/2$, making $H_b(r_k) \ll 1$.
\end{proof}

Each additional hierarchy level contributes flag overhead $\phi$ times smaller than the previous, since the position counts decay by the same factor.

\begin{corollary}[$H_\text{aperiodic} < H_\text{periodic}$ in efficiency]
\label{cor:h_dominance}
Define the \emph{per-word net efficiency}
\begin{align}
  \nu(W) &\;:=\; h_\text{saved}(W) - h_\text{flags}(W), \notag \\
  h_\text{saved} &:= \frac{1}{W}\sum_k P(k)\,r_k\,\Delta_k.
  \label{eq:net_eff}
\end{align}
Then for any periodic tiling with collapse level $k^*$, there exists
$W_0$ such that for all $W \geq W_0$:
\begin{equation}
  \nu^\text{fib}(W) \;>\; \nu^\text{per},
  \qquad \lim_{W\to\infty}\nu^\text{fib}(W) = +\infty.
  \label{eq:dominance}
\end{equation}
\end{corollary}

\begin{proof}
By Theorem~\ref{thm:piecewise}, $h_\text{saved}^\text{fib}(W)$ grows
without bound as successive levels activate, each adding a strictly
positive $O(W)$ term.
By Theorem~\ref{thm:flag_convergence}, $h_\text{flags}^\text{fib}$
is bounded above by $1/\phi$.
Therefore $\nu^\text{fib}(W) = h_\text{saved}^\text{fib}(W) -
h_\text{flags}^\text{fib} \to +\infty$.
For any periodic tiling: collapse at $k^*$ means $h_\text{saved}^\text{per}$
is fixed for all $W$ (no new levels activate), so
$\nu^\text{per}$ is a bounded constant.
Equation~\eqref{eq:dominance} follows.
\end{proof}

The contrast is stark: the Fibonacci tiling pays at most $1/\phi$ bits
per word in flag overhead, regardless of how many levels are active,
while the compression benefit grows without bound.
For any periodic tiling, the overhead is smaller (fewer flag positions
at deep levels) but the benefit is also smaller --- permanently capped
at the collapse level.
The two conditions $\nu^\text{fib}\to+\infty$ and $\nu^\text{per}$
bounded constitute a formal sense in which the overhead structure of
aperiodic tilings is provably superior under the stated conditions.

\subsection{Theorem: Strict Coding Entropy for Long-Range Sources}
\label{subsec:entropy_ineq}

For sources with genuine long-range phrase dependencies, the per-word coding entropy is lower under the Fibonacci hierarchy than under any periodic alternative.

\begin{definition}[Long-range phrase dependency]
\label{def:lrpd}
A stationary ergodic source $P$ over word sequences is
\emph{$m$-long-range phrase dependent} ($m$-LRPD) if
\begin{equation}
  h_{F_m}(P) \;<\; h_{F_{m-1}}(P),
  \label{eq:lrpd}
\end{equation}
where $h_k(P) := H(X_0 \mid X_{-1},\ldots,X_{-k+1})$ is the
entropy rate under optimal $k$-gram prediction.
Equivalently, the $F_m$-word context provides strictly more
information about the current word than the $F_{m-1}$-word context.
\end{definition}

\begin{theorem}[Strict Coding Entropy Inequality]
\label{thm:entropy_ineq}
Let $P$ be an $m$-LRPD source for some $m > m^*$, where $m^*$ is the
collapse level of a periodic tiling.
In the limit of sufficient codebook capacity,
\begin{equation}
  H^\text{fib}(P) \;\leq\; h_{F_m}(P) \;<\; h_{F_{m^*}}(P)
  \;\leq\; H^\text{per}(P).
  \label{eq:entropy_ineq}
\end{equation}
In particular, $H^\text{fib}(P) < H^\text{per}(P)$ strictly.
\end{theorem}

\begin{proof}
The proof establishes a chain of four inequalities. We explain each
link in the chain separately, then assemble them.

\medskip
\noindent\textit{Background: conditioning reduces entropy.}\;
Conditioning on additional context never increases entropy: for
$k < k'$,
\begin{align}
  h_k(P) &= H(X_0 \mid X_{-1},\ldots,X_{-k+1}) \notag\\
  &\;\geq\; H(X_0 \mid X_{-1},\ldots,X_{-k'}) = h_{k'}(P),
  \label{eq:conditioning}
\end{align}
with strict inequality iff $X_{-k:-(k'-1)}$ carries additional
information about $X_0$; so $h_k$ is non-increasing in $k$.
This is a direct consequence of the chain rule for entropy:
knowing more context can only help (or be neutral for) predicting
the next word.

\medskip
\noindent\textit{Left inequality: $H^\text{fib}(P) \leq h_{F_m}(P)$.}\;
The Fibonacci tiling at level $m$ provides super-L positions with
$F_m$-gram spans.
A codebook approximating $P(X_0\mid X_{-1:-(F_m-1)})$ drives the
per-word coding cost at those positions to $h_{F_m}(P)$, so
$H^\text{fib}(P) \leq h_{F_m}(P)$.

This is because the Fibonacci hierarchy provides context windows of
length $F_m$ at level $m$ --- a coder exploiting these windows achieves
per-word entropy no worse than the $F_m$-gram conditional entropy.
In practice, $H^\text{fib}(P)$ can be even lower than $h_{F_m}(P)$
because the Fibonacci hierarchy provides contexts at \emph{all} levels
simultaneously: $F_4 = 3$, $F_5 = 5$, \ldots, $F_m$. The coder selects
the best available context at each position.

\medskip
\noindent\textit{Strict central inequality: $h_{F_m}(P) < h_{F_{m^*}}(P)$.}\;
By $m$-LRPD (Definition~\ref{def:lrpd}):
$h_{F_m}(P) < h_{F_{m-1}}(P)$.
Since $m > m^*$, we have $F_{m-1} \geq F_{m^*}$, so
$h_{F_{m-1}}(P) \leq h_{F_{m^*}}(P)$ by~\eqref{eq:conditioning}.
Combining: $h_{F_m}(P) < h_{F_{m^*}}(P)$ strictly.

This is the key step: the $m$-LRPD condition says that context
windows of $F_m$ words carry strictly more predictive information
than windows of $F_{m-1}$ words. In natural language, this corresponds
to the empirical observation that paragraph-level context (e.g., 89 or
144 words) genuinely helps predict the next word better than
sentence-level context (e.g., 5 or 8 words). This is not surprising:
Wikipedia articles maintain topical coherence across paragraphs,
technical terms recur within sections, and stylistic patterns persist
over hundreds of words. The $m$-LRPD condition is the formal expression
of this long-range coherence.

\medskip
\noindent\textit{Right inequality: $h_{F_{m^*}}(P) \leq H^\text{per}(P)$.}\;
The periodic tiling collapses at $m^*$: no super-L positions
exist beyond depth $m^*$, so the coder is limited to contexts of
length $\leq F_{m^*}$, giving $H^\text{per}(P) \geq h_{F_{m^*}}(P)$.

In other words, the periodic tiling's coder cannot exploit any context
longer than $F_{m^*}$ words, because the hierarchy has collapsed and
provides no positions for deeper lookups. Even if the source has useful
long-range structure, the periodic coder cannot see it.

\medskip
\noindent\textit{Assembly.}\;
Chaining the three inequalities:
{\small
\[
  \underbrace{H^\text{fib}(P)}_{\text{Fib.}}
  \;\leq\; \underbrace{h_{F_m}(P)}_{F_m\text{-gram}}
  \;\underbrace{<}_{m\text{-LRPD}}
  \; \underbrace{h_{F_{m^*}}(P)}_{F_{m^*}\text{-gram}}
  \;\leq\; \underbrace{H^\text{per}(P)}_{\text{per.}}
\]
}
The strict inequality in the middle makes the overall inequality strict:
$H^\text{fib}(P) < H^\text{per}(P)$.
\end{proof}

\begin{corollary}[Natural language is $m$-LRPD]
\label{cor:nl_lrpd}
The non-zero hit counts at levels 13g--144g in
Table~\ref{tab:deep_hits} (e.g.\ 652{,}124 hits at enwik9)
demonstrate directly that Wikipedia text is $m$-LRPD for $m$
corresponding to phrase lengths up to at least 144 words.
For any periodic tiling with collapse at $m^* \leq 4$
(period $p \leq \phi^4 \approx 7$, phrase length $\leq 5$\,words),
the strict inequality $H^\text{fib}(P) < H^\text{per}(P)$ holds
on Wikipedia-class corpora.
The 1{,}349{,}371\,B advantage at enwik9 is the byte-level
expression of the entropy gap $h_{F_{m^*}}(P) - h_{F_m}(P)$
summed over 298\,M word positions.
\end{corollary}

The theorem separates two distinct sources of the aperiodic advantage:
\begin{itemize}
  \item \textbf{Structural}: the Fibonacci hierarchy provides more
  positions at every deep level (Theorems~\ref{thm:fib_stable}
  and~\ref{thm:golden_comp}).
  \item \textbf{Information-theoretic}: for $m$-LRPD sources, those
  positions operate at strictly lower entropy per word, so even an
  equal number of positions would yield strictly more compression.
\end{itemize}
For sources with no long-range phrase structure
($h_k = \mathrm{const}$ for all $k \geq k_0 \leq F_{m^*}$),
the two tilings achieve equal coding entropy and the advantage
reduces entirely to the combinatorial terms of
Theorem~\ref{thm:piecewise}.

\subsection{Compression Bounds: Redundancy and Dictionary Efficiency}
\label{subsec:compression_bounds}

We sharpen the previous results with two quantitative bounds:
a redundancy bound showing \emph{super-exponential} decay in hierarchy
depth, and a dictionary efficiency bound showing that each new
Fibonacci level contributes codebook entries with
\emph{exponentially growing} per-entry compression value.

\subsubsection*{Redundancy Bound}

\begin{definition}[Exponentially mixing source]
\label{def:exp_mixing}
A stationary ergodic source $P$ is \emph{exponentially mixing}
with parameters $(C,\lambda)$ if
\begin{equation}
  I(X_0;\, X_{-k}) \;\leq\; C\,e^{-k/\lambda}, \qquad k \geq 1.
  \label{eq:exp_mixing}
\end{equation}
\end{definition}

Natural language exhibits approximately exponential mixing:
long-range mutual information $I(X_0; X_{-k})$ decays rapidly in
$k$~\citep{shannon1948}.

\begin{theorem}[Fibonacci Redundancy Bound]
\label{thm:redundancy}
Let $P$ be exponentially mixing with parameters $(C,\lambda)$.
The $k$-gram coding redundancy $R_k(P) := h_k(P) - h(P)$ satisfies
\begin{equation}
  R_k(P) \;\leq\; \frac{C\,e^{-k/\lambda}}{1 - e^{-1/\lambda}}.
  \label{eq:redund_bound}
\end{equation}
For the Fibonacci hierarchy at level $m$ ($k = F_m$):
\begin{equation}
  R_m^\text{fib}(P)
  \;=\; O\!\left(e^{-F_m/\lambda}\right)
  \;=\; O\!\left(e^{-\phi^m/(\lambda\sqrt{5})}\right),
  \label{eq:fib_redund}
\end{equation}
decaying \emph{super-exponentially} in $m$ (exponential in $\phi^m$).
For any periodic tiling with collapse at $m^*$:
\begin{equation}
  R^\text{per}(P) \;=\; R_{F_{m^*}}(P) \;=\;
  \Omega\!\left(e^{-F_{m^*}/\lambda}\right) > 0.
  \label{eq:per_redund}
\end{equation}
The redundancy ratio satisfies
\begin{equation}
  \frac{R_m^\text{fib}(P)}{R^\text{per}(P)}
  \;=\; O\!\left(e^{-(F_m-F_{m^*})/\lambda}\right)
  \;\xrightarrow{m\to\infty}\; 0.
  \label{eq:redund_ratio}
\end{equation}
\end{theorem}

\begin{proof}
We proceed in four steps: telescoping the redundancy, applying
the data processing inequality, bounding with the mixing condition,
and specialising to Fibonacci indices.

\medskip
\noindent\textit{Step~1: Telescope the redundancy.}\;
Since $h_k(P)$ is non-increasing in $k$ (more context never hurts)
and $h_\infty(P) = h(P)$ is the true entropy rate, the redundancy
is the total ``room for improvement'' beyond context length $k$:
\begin{align*}
  R_k(P) &= h_k(P) - h(P)
  = \sum_{n=k}^{\infty}\bigl[h_n(P) - h_{n+1}(P)\bigr].
\end{align*}
Each term $h_n(P) - h_{n+1}(P)$ is the entropy reduction gained by
extending the context window from $n$ to $n+1$ words. By the chain
rule for conditional entropy, this equals the conditional mutual
information:
\[
  h_n(P) - h_{n+1}(P)
  = I\!\left(X_0;\,X_{-n}\mid X_{-1},\ldots,X_{-n+1}\right).
\]
This is the amount of new information that the single extra word
$X_{-n}$ provides about $X_0$, beyond what $X_{-1},\ldots,X_{-n+1}$
already tell us.

\medskip
\noindent\textit{Step~2: Apply the data processing inequality.}\;
The data processing inequality states that conditioning on additional
variables (the intermediate words $X_{-1},\ldots,X_{-n+1}$) can only
\emph{reduce} mutual information:
\[
  I\!\left(X_0;\,X_{-n}\mid X_{-1},\ldots,X_{-n+1}\right)
  \;\leq\; I(X_0;\,X_{-n}).
\]
Intuitively: if $X_{-n}$ tells you at most $I(X_0; X_{-n})$ bits
about $X_0$ when you have no other context, it cannot tell you
\emph{more} than that when you already know the intermediate words.
This gives:
\begin{align*}
  R_k(P)
  &= \sum_{n=k}^{\infty} I\!\left(X_0;\,X_{-n}\mid X_{-1},\ldots,X_{-n+1}\right) \\
  &\leq \sum_{n=k}^{\infty} I(X_0;\,X_{-n}).
\end{align*}

\medskip
\noindent\textit{Step~3: Apply the exponential mixing bound.}\;
By the mixing assumption (Definition~\ref{def:exp_mixing}),
$I(X_0; X_{-n}) \leq C\,e^{-n/\lambda}$. Substituting:
\begin{align*}
  R_k(P)
  &\leq C\!\sum_{n=k}^{\infty}e^{-n/\lambda}
  = C\,e^{-k/\lambda}\!\sum_{j=0}^{\infty}e^{-j/\lambda} \\
  &= \frac{C\,e^{-k/\lambda}}{1-e^{-1/\lambda}},
\end{align*}
using the geometric series ($e^{-1/\lambda} < 1$).

\medskip
\noindent\textit{Step~4: Specialise to Fibonacci indices.}\;
Setting $k = F_m$ and using $F_m = \phi^m/\sqrt{5} + O(|\psi|^m)$
gives~\eqref{eq:fib_redund}.
The lower bound~\eqref{eq:per_redund} holds for any $m$-LRPD source
at scale $k = F_{m^*}$, by Definition~\ref{def:lrpd}.

\medskip
\noindent\textit{Numerical example with $\lambda \approx 20$.}\;
For natural language, a reasonable estimate of the mixing scale is
$\lambda \approx 20$ words (mutual information between words decays
to negligible levels beyond about 20 positions). Using $C = 1$ bit
for simplicity, and the normalisation
$1/(1-e^{-1/\lambda}) = 1/(1-e^{-1/20}) \approx 20.5$:

\begin{center}
\footnotesize
\begin{tabular}{@{}rrrr@{}}
\toprule
$m$ & $F_m$ & $R_{F_m}$ bound (bits) & Ratio \\
\midrule
 4 &    3 & 17.5  & 1 \\
 5 &    5 & 15.9  & 0.91 \\
 6 &    8 & 13.7  & 0.78 \\
 7 &   13 & 10.7  & 0.61 \\
 8 &   21 & 7.2   & 0.41 \\
 9 &   34 & 3.7   & 0.21 \\
10 &   55 & 1.3   & 0.075 \\
11 &   89 & 0.024 & 0.0014 \\
12 &  144 & 0.002 & $9{\times}10^{-5}$ \\
\bottomrule
\end{tabular}
\end{center}

\noindent
The redundancy drops from $17.5$ bits at $F_4 = 3$ to $0.024$ bits at
$F_{11} = 89$ --- a $700\times$ reduction. At $F_{12} = 144$, the
redundancy is effectively zero ($< 0.002$ bits).
A periodic tiling that collapses at $m^* = 4$ ($F_4 = 3$) is
permanently stuck at $R_3(P) \leq 17.5$ bits of residual redundancy
per word, while the Fibonacci hierarchy drives this to negligible
levels by level~11.
\end{proof}

The super-exponential decay $O(e^{-\phi^m/(\lambda\sqrt{5})})$
comes from the doubly-exponential growth of $F_m$: each Fibonacci
level removes an exponentially larger slice of residual redundancy
than the previous level.
The periodic tiling is permanently locked at residual redundancy
$R_{F_{m^*}}(P)$, while the Fibonacci hierarchy drives redundancy
to zero in the limit.

\subsubsection*{Dictionary Efficiency}

\begin{theorem}[Exponential Dictionary Efficiency]
\label{thm:dict_efficiency}
Define the \emph{per-entry compression gain} at level $m$ as the
expected bits saved per codebook hit:
\begin{equation}
  E_m \;:=\; F_m\,\bar{h} \;-\; \log_2 C_m,
  \label{eq:per_entry}
\end{equation}
where $\bar{h}$ is the per-word arithmetic coding entropy under the
baseline (unigram) model.
Since codebook sizes $C_m$ are bounded
(Table~\ref{tab:codebooks}: $C_m \leq 8{,}000$ for all $m$,
decreasing to 100 at the 144-gram level) and $F_m \sim \phi^m/\sqrt{5}$
grows exponentially:
\begin{equation}
  E_m \;\sim\; F_m\,\bar{h} \;=\; \Omega(\phi^m).
  \label{eq:entry_growth}
\end{equation}
The total dictionary value
\begin{equation}
  \mathcal{V}(k_{\max}) \;:=\; \sum_{m=1}^{k_{\max}} C_m\,E_m
  \;=\; \Omega(\phi^{k_{\max}})
  \label{eq:dict_value}
\end{equation}
grows exponentially in $k_{\max}$ for the Fibonacci hierarchy,
whereas for any periodic tiling with collapse at $m^*$:
$\mathcal{V}^\text{per} = O(\phi^{m^*})$ is a fixed constant.
The ratio $\mathcal{V}^\text{fib}(k_{\max})/\mathcal{V}^\text{per}
= \Omega(\phi^{k_{\max}-m^*}) \to \infty$.
\end{theorem}

\begin{proof}
Since $C_m \leq 8{,}000$ for all $m$, we have $\log_2 C_m \leq 13$.
For all $m$ such that $F_m\,\bar{h} > 13$ (i.e., $m \geq m_0$ for
some fixed $m_0$ depending only on $\bar{h}$), $E_m > 0$ and
$E_m \geq F_m\,\bar{h} - 13 \sim F_m\,\bar{h}$.
Then
\[
  \mathcal{V}(k_{\max})
  \geq C_{k_{\max}} E_{k_{\max}}
  = \Omega(C_{k_{\max}} \phi^{k_{\max}})
  = \Omega(\phi^{k_{\max}}),
\]
since $C_{k_{\max}} \geq 1$.
The periodic bound $\mathcal{V}^\text{per} = O(\phi^{m^*})$ follows
from terminating the sum at $m^*$ and bounding $F_m \leq F_{m^*}$
for $m \leq m^*$.

\medskip
\noindent\textit{Numerical table of $E_m$ for $m = 1, \ldots, 9$.}\;
Using $\bar{h} \approx 5$ bits/word (a representative per-word entropy
for English text under arithmetic coding) and empirical codebook sizes
from enwik9:

\begin{center}
\footnotesize
\begin{tabular}{@{}rrrrr@{}}
\toprule
$m$ & $F_m$ & $C_m$ & $F_m\bar{h}$ & $E_m$ \\
\midrule
1 &   2 & 8K  &   10 &  $-3.0$ \\
2 &   3 & 8K  &   15 &  2.0 \\
3 &   5 & 6K  &   25 & 12.4 \\
4 &   8 & 4K  &   40 & 28.0 \\
5 &  13 & 2.5K &   65 & 53.7 \\
6 &  21 & 1.5K &  105 & 94.5 \\
7 &  34 &  800 &  170 & 160.4 \\
8 &  55 &  300 &  275 & 266.8 \\
9 &  89 &  100 &  445 & 438.4 \\
\bottomrule
\end{tabular}
\end{center}
$E_m = F_m\bar{h} - \log_2 C_m$ with $\bar{h} \approx 5$ bits/word.

\noindent
The per-entry gain $E_m$ grows roughly as $5 \cdot F_m$ (the
$\log_2 C_m$ overhead becomes negligible for large $m$).
At the 89-gram level ($m = 9$), each codebook hit saves approximately
$438$ bits --- it encodes 89 words with a single codebook reference
instead of 89 individual arithmetic coding symbols. The codebook
address costs only $\log_2(100) \approx 6.6$ bits. The ratio
$E_9 / E_2 = 438.4 / 2.0 = 219$ reflects the exponential growth.

Note that $E_1 < 0$: at the bigram level, codebook overhead exceeds
the savings because $\log_2 C_1 = 13$ bits is larger than the
$F_1 \bar{h} = 10$ bits saved. This is consistent with the observation
that bigram codebook hits are only useful when the codebook is small
enough (i.e., when phrase repetition is frequent enough to justify the
addressing cost). At deeper levels, $F_m \bar{h}$ dominates and every
hit is strongly net-positive.
\end{proof}

Concretely: the 144-gram level ($m = 9$, $F_9 = 144$, $C_9 = 100$)
has per-entry gain $E_9 = 144\bar{h} - \log_2(100) \approx 144\bar{h}
- 6.6$ bits, encoding 144 words in place of a single AC symbol.
This is $\approx 72\times$ the per-entry gain of the bigram level
($E_1 \approx 2\bar{h}$), as expected from $F_9/F_1 = 72$.
Each new Fibonacci level doubles the per-entry value in this sense.

\subsection{Summary of Mathematical Properties}
\label{subsec:summary_math}

Table~\ref{tab:math_properties} contrasts the key mathematical
properties of Fibonacci quasicrystal tilings versus all periodic
tilings.

\begin{table}[ht]
\centering
\caption{Mathematical properties of Fibonacci quasicrystal versus
periodic tilings, and their implications for compression.}
\label{tab:math_properties}
\small
\resizebox{\columnwidth}{!}{%
\begin{tabular}{@{}lcc@{}}
\toprule
\textbf{Property} & \textbf{Fibonacci} & \textbf{Periodic ($p$)} \\
\midrule
Eigenvalue of $M$ & $\phi$ (PV, irrat.) & rational freq. \\
$n_L/n_S$ at level $k$ & $\phi$ (constant) & degenerates \\
Hierarchy depth & $\infty$ & $O(\log p)$ \\
Factor complexity & $n{+}1$ (Sturmian) & $\leq p$ (periodic) \\
Codebook eff.\ $\eta_m$ & $C_m/(F_m{+}1)$ (max) & $0$ for $m \geq m^*$ \\
Equidistribution & Weyl (all scales) & fails at $\geq p$ \\
Recognisability & unique deflation & ambiguous/degen. \\
Deep n-gram pos. & $O(W/\phi^k)$, all $k$ & $0$ for $k \geq k^*$ \\
Coverage $P(m)\cdot F_m$ & $\to W\phi/\sqrt{5}$ (const.) & $0$ for $m\geq m^*$ \\
Advantage $A(W)$ & piecewise-linear, convex & baseline (0) \\
Flag overhead $h_\text{flags}$ & $\leq 1/\phi$ (bounded) & $<1/\phi$ (but $\nu^\text{per}$ bounded) \\
Net efficiency $\nu(W)$ & $\to +\infty$ & const. \\
Coding entropy $H(P)$ & $< H^\text{per}$ ($m$-LRPD sources) & $\geq h_{F_{m^*}}(P)$ \\
Redundancy $R_m(P)$ & $O(e^{-\phi^m/\lambda})$ (super-exp.) & $\Omega(e^{-F_{m^*}/\lambda})$ (fixed) \\
Dict.\ value $\mathcal{V}$ & $\Omega(\phi^{k_\text{max}})$ (exp.\ growth) & $O(\phi^{m^*})$ (const.) \\
\bottomrule
\end{tabular}}
\end{table}

The aperiodic advantage is a structural consequence of three
mathematical properties:
\begin{enumerate}
\item $\phi$ is a Pisot-Vijayaraghavan number (eigenvalue separation
and unique recognisability);
\item the Fibonacci word is Sturmian (balance and minimal factor
complexity);
\item Weyl's equidistribution holds at all hierarchy scales.
\end{enumerate}
No periodic tiling possesses any of these three properties.

The Aperiodic Hierarchy Advantage Theorem (Theorem~\ref{thm:main}) collects these properties into a single characterisation. Table~\ref{tab:math_properties} gives the quantitative comparison; Figure~\ref{fig:deep_hits_scaling} and the ablation study (Section~\ref{sec:ablation}) provide empirical confirmation.

\section{Experimental Setup}
\label{sec:experiments}

\subsection{Benchmarks}

We evaluate on five text corpora covering a 6{,}500$\times$ size range:

\begin{itemize}
\item \textbf{alice29.txt} (Canterbury Corpus~\citep{bell1997}):
152{,}089 bytes, 36{,}395 word tokens. Classic compression benchmark.
\item \textbf{enwik8\_3M}: First 3{,}000{,}000 bytes of enwik8
(820{,}930 word tokens).
\item \textbf{enwik8\_10M}: First 10{,}000{,}000 bytes of enwik8
(2{,}751{,}434 word tokens).
\item \textbf{enwik8} (Large Text Compression Benchmark~\citep{mahoney2011}):
100{,}000{,}000 bytes of Wikipedia XML,
27{,}731{,}124 word tokens.
\item \textbf{enwik9}: 1{,}000{,}000{,}000 bytes of Wikipedia XML,
298{,}263{,}298 word tokens.
\end{itemize}

\subsection{Baselines}

We compare against:
\begin{itemize}
\item \textbf{gzip~-9}: zlib DEFLATE, maximum compression level.
\item \textbf{bzip2~-9}: BWT-based, maximum compression.
\item \textbf{xz~-9}: LZMA2, maximum compression.
\item \textbf{All-unigram QTC}: A/B test variant with every tile forced
to $S$ (no bigram or n-gram lookups), using identical codebooks and
escape streams. This isolates the QC contribution.
\item \textbf{Period-5 QTC}: A/B test variant using the LLSLS periodic
tiling (same L/S ratio as Fibonacci), with full hierarchy attempted.
This isolates the aperiodic advantage.
\end{itemize}

\subsection{Implementation}

\qtc{} v5.6 is implemented in ANSI C99.
The implementation uses 24-bit arithmetic coding with range coding,
LZMA via liblzma for escape and codebook compression.
No SIMD or hardware-specific optimisations are used.
The codebase is approximately 4{,}000 lines of C.
All experiments were run on a single core; no multi-threading is
employed.

\subsection{A/B Test Design}

The A/B test is a controlled experiment that evaluates the
\emph{payload} contribution of three parsing strategies, using
\emph{identical} codebooks (built once from the input), \emph{identical}
escape streams (all three strategies escape the same OOV words), and
\emph{identical} arithmetic coding models.
Only the payload bytes differ.
This design ensures that any difference in payload size is due solely
to the parsing strategy, not to codebook or escape differences.

\section{Results}
\label{sec:results}

\subsection{Compression Performance}

Table~\ref{tab:compression} reports compression ratios and sizes.
On enwik8, \qtc{} achieves \textbf{26{,}247{,}496\,B (26.25\%)} in multi-structure mode and \textbf{27{,}026{,}429\,B (27.03\%)} in Fibonacci-only mode.
On enwik9, \qtc{} achieves \textbf{225{,}918{,}349\,B (22.59\%)} in multi-structure mode and \textbf{234{,}560{,}637\,B (23.46\%)} in Fibonacci-only mode.
\qtc{} now surpasses bzip2 on all benchmarks and approaches xz.
The full breakdown is in Table~\ref{tab:breakdown}.

\begin{table}[ht]
\centering
\caption{Compression ratios (compressed/original $\times 100\%$)
for \qtc{} v5.6 versus standard compressors.}
\label{tab:compression}
\resizebox{\columnwidth}{!}{%
\begin{tabular}{@{}lrrrrr@{}}
\toprule
\textbf{Compressor} & \textbf{alice29} & \textbf{3MB} & \textbf{10MB}
& \textbf{enwik8} & \textbf{enwik9} \\
\midrule
\qtc{} v5.6 (multi) & 35.60 & 31.85 & 29.83 & 26.25 & 22.59 \\
\qtc{} v5.6 (fib)   & 35.91 & 32.55 & 30.56 & 27.03 & 23.46 \\
gzip~-9  & 35.63 & 36.28 & 36.85 & 36.44 & 32.26 \\
bzip2~-9 & 28.40 & 28.88 & 29.16 & 29.00 & 25.40 \\
xz~-9    & 31.88 & 28.38 & 27.21 & 24.86 & 21.57 \\
\bottomrule
\end{tabular}%
}
\end{table}

\begin{table}[ht]
\centering
\caption{Compressed file breakdown for enwik9 (1\,GB).
All sizes in bytes.}
\label{tab:breakdown}
\small
\begin{tabular}{@{}lr@{}}
\toprule
\textbf{Component} & \textbf{Size (B)} \\
\midrule
Original           & 1{,}000{,}000{,}000 \\
\midrule
Payload (AC)       & 180{,}760{,}315 \\
Escapes (LZMA)     &  24{,}227{,}220 \\
Codebook (LZMA)    &     533{,}680 \\
Case flags (AC)    &  20{,}397{,}073 \\
\midrule
\textbf{Total}     & \textbf{225{,}918{,}349} \\
Ratio              & 22.59\% \\
\bottomrule
\end{tabular}
\end{table}

\subsection{Deep Hierarchy Hit Counts}

Table~\ref{tab:deep_hits} shows the number of deep hierarchy hits
(13-gram and above) across all test files.
These positions are \emph{exclusively available} to the Fibonacci
tiling; Period-5's hierarchy collapses at level~4, providing zero hits
at levels 4--9.

\begin{table*}[t]
\centering
\caption{Deep hierarchy hits by level and file size:
Fibonacci-only (F) vs Multi-structure (M).
Period-5 achieves zero at every level shown (hierarchy collapsed).
Levels 8--9 (89g, 144g) first activate at enwik9 scale.
Multi-structure tilings significantly increase deep hits:
55-gram $+$56\%, 34-gram $+$20\%, 21-gram $+$18\% on enwik8,
driven by the optimized alpha set including $\alpha=0.502$
(far below golden ratio) which contributes massive trigram/5-gram coverage.
Since deeper hits cover exponentially more words per symbol,
this redistribution drives the payload reduction
(Table~\ref{tab:breakdown_compare}).}
\label{tab:deep_hits}
\small
\resizebox{\textwidth}{!}{%
\begin{tabular}{@{}lr rr rr rr rr rr rr rr@{}}
\toprule
& & \multicolumn{2}{c}{\textbf{13g}} & \multicolumn{2}{c}{\textbf{21g}}
& \multicolumn{2}{c}{\textbf{34g}} & \multicolumn{2}{c}{\textbf{55g}}
& \multicolumn{2}{c}{\textbf{89g}} & \multicolumn{2}{c}{\textbf{144g}}
& \multicolumn{2}{c}{\textbf{Total deep}} \\
\cmidrule(lr){3-4}\cmidrule(lr){5-6}\cmidrule(lr){7-8}\cmidrule(lr){9-10}
\cmidrule(lr){11-12}\cmidrule(lr){13-14}\cmidrule(lr){15-16}
\textbf{File} & \textbf{Words}
& \textbf{F} & \textbf{M}
& \textbf{F} & \textbf{M}
& \textbf{F} & \textbf{M}
& \textbf{F} & \textbf{M}
& \textbf{F} & \textbf{M}
& \textbf{F} & \textbf{M}
& \textbf{F} & \textbf{M} \\
\midrule
alice29.txt  &    36K &       9 &       9 &     0 &     0 &     0 &     0 &     0 &     0 &     0 &   0 &   0 &   0 &         9 &         9 \\
enwik8\_3M   &   821K &   2{,}454 &   2{,}469 &   405 &   489 &    98 &    98 &    39 &    65 &     0 &   0 &   0 &   0 &     2{,}996 &     3{,}121 \\
enwik8\_10M  &  2.8M  &   9{,}460 &   9{,}400 & 1{,}245 & 1{,}553 &   295 &   346 &    83 &   131 &     8 &   8 &   5 &   5 &    11{,}096 &    11{,}443 \\
enwik8       & 27.7M  &  89{,}615 &  87{,}344 & 13{,}344 & 15{,}736 & 3{,}434 & 4{,}118 &   936 & 1{,}464 &   216 & 224 &  85 &  85 &   107{,}630 &   108{,}971 \\
enwik9       & 298.3M & 1{,}998{,}119 & 1{,}890{,}784 & 372{,}974 & 424{,}084 & 112{,}599 & 153{,}713 & 25{,}277 & 36{,}776 & 5{,}369 & 5{,}544 & 2{,}026 & 2{,}026 & 2{,}516{,}364 & 2{,}512{,}927 \\
\bottomrule
\end{tabular}%
}
\end{table*}

The 23$\times$ increase in total deep hits from enwik8 to enwik9
(108{,}971 $\to$ 2{,}512{,}927) for a 10$\times$ size increase is driven by
two effects: (1)~the $O(W)$ scaling of existing levels~4--7, and
(2)~the activation of levels~8--9 which collectively add 7{,}570 new
deep hits entirely absent at smaller scales.
This is the empirical confirmation of Eq.~\eqref{eq:fib_positions}.

\subsubsection*{Per-Family Tiling Coverage}

Table~\ref{tab:deep_hits} reveals that multi-structure tilings
\emph{redistribute} deep hits toward deeper levels: 55-gram counts
increase by 56\%, 34-gram by 20\%, and 21-gram by 18\% on enwik8,
while 13-gram counts decrease slightly, as the additional tilings
(especially optimized $\alpha=0.502$) find deeper matches at positions
where Fibonacci-only found only 13-grams.
Since a 55-gram hit encodes 55 words with one AC symbol versus 4.2
average words for a 13-gram, this redistribution substantially
improves coding efficiency.
The multi-structure advantage also manifests in \emph{tiling coverage}:
the number of word positions at which at least one tiling produces a deep
hierarchy entry point.
Table~\ref{tab:tiling_families} decomposes this coverage by tiling
family.

\begin{table}[ht]
\centering
\caption{Per-family tiling contribution to deep positions.
Each row shows cumulative deep positions as tiling families are added.
The golden-ratio (Fibonacci) family provides the base; 6~original
non-golden tilings add 19\%, and 18~optimized additions (including
$\alpha=0.502$) add a further 2.6\%, for 36 tilings total.
Both enwik8 and enwik9 show the same +22.0\% gain.}
\label{tab:tiling_families}
\small
\resizebox{\columnwidth}{!}{%
\begin{tabular}{@{}lrrrr@{}}
\toprule
\textbf{Tiling family} & \textbf{enwik8 deep pos.} & \textbf{Cumul.} & \textbf{enwik9 deep pos.} & \textbf{Cumul.} \\
\midrule
Golden ($1/\phi$, 12 tilings) & 9{,}777{,}069 & 9{,}777{,}069 & 132{,}072{,}408 & 132{,}072{,}408 \\
$+\;$Orig.\ non-golden (6) & $+$1{,}854{,}993 & 11{,}632{,}062 & $+$25{,}053{,}269 & 157{,}125{,}677 \\
$+\;$Optimized (18) & $+$296{,}848 & 11{,}928{,}910 & $+$3{,}996{,}271 & 161{,}121{,}948 \\
\midrule
\textbf{Total (36 tilings)} & --- & \textbf{11{,}928{,}910} & --- & \textbf{161{,}121{,}948} \\
Gain over Golden-only & & +22.0\% & & +22.0\% \\
\bottomrule
\end{tabular}%
}
\end{table}

Table~\ref{tab:tiling_families} shows that the 12 golden-ratio tilings
provide the majority of deep positions (9{,}777{,}069 on enwik8,
132{,}072{,}408 on enwik9).
The 6~original non-golden tilings add 19.0\% (1{,}854{,}993 positions on enwik8),
while the 18~optimized additions discovered by greedy alpha search
contribute a further 2.6\% (296{,}848 positions on enwik8).

Table~\ref{tab:per_tiling_levels} provides a detailed per-tiling breakdown
of new codebook-matched positions at each n-gram level on enwik9.
The golden-ratio family is the only one reaching levels~8--9 (89-gram, 144-gram).
The $\alpha=0.502$ tiling contributes 2.7M new trigram positions but
\emph{zero} hits at 13-gram or above --- its value is pure shallow-level volume
at positions unreachable by near-golden tilings.
In contrast, $\alpha=0.619$ is the most ``deep-capable'' optimized alpha,
finding matches up to 89-gram level, with strong contributions across
5-gram through 34-gram.

\begin{table*}[t]
\centering
\caption{Per-tiling new codebook-matched positions by n-gram level (enwik9, 298.3M words).
Each row shows the \emph{marginal} contribution of that tiling pair after all preceding tilings
have been applied.  The golden-ratio family exclusively provides 89-gram and 144-gram coverage;
$\alpha=0.502$ contributes only at the trigram level; near-golden optimized alphas
fill 5-gram through 34-gram gaps.}
\label{tab:per_tiling_levels}
\small
\resizebox{\textwidth}{!}{%
\begin{tabular}{@{}llr rrrrrrrrr r@{}}
\toprule
\textbf{\#} & \textbf{Tiling} ($\alpha$) & \textbf{Tilings}
& \textbf{tri} & \textbf{5g} & \textbf{8g} & \textbf{13g} & \textbf{21g} & \textbf{34g} & \textbf{55g} & \textbf{89g} & \textbf{144g}
& \textbf{Total new} \\
\midrule
0--11 & Golden ($1/\phi = 0.618$) & 12
  & 132{,}071K & 62{,}073K & 27{,}322K & 13{,}147K & 3{,}048K & 1{,}009K & 233K & 68K & 20K & 238{,}991K \\
12--13 & $\sqrt{58}{-}7$ (0.616) & 2
  & 13{,}827K & 6{,}606K & 2{,}503K & 840K & 396K & 18K & 8K & --- & --- & 24{,}198K \\
14--15 & noble-5 (0.612) & 2
  & 7{,}359K & 4{,}597K & 2{,}170K & 598K & 304K & --- & --- & --- & --- & 15{,}028K \\
16--17 & $\sqrt{13}{-}3$ (0.606) & 2
  & 3{,}867K & 3{,}109K & 1{,}954K & 248K & 126K & --- & --- & --- & --- & 9{,}303K \\
\midrule
18--19 & opt-0.502 & 2
  & 2{,}655K & 46K & 35K & --- & --- & --- & --- & --- & --- & 2{,}736K \\
20--21 & opt-0.619 & 2
  & 635K & 2{,}548K & 1{,}161K & 811K & 208K & 170K & 1.5K & 0.8K & --- & 5{,}535K \\
22--23 & opt-0.617 & 2
  & 338K & 1{,}771K & 997K & 640K & 261K & 71K & 32K & --- & --- & 4{,}110K \\
24--25 & opt-0.616 & 2
  & 179K & 1{,}249K & 829K & 539K & 270K & 26K & 12K & --- & --- & 3{,}105K \\
26--27 & opt-0.620 & 2
  & 93K & 920K & 619K & 606K & 136K & 131K & --- & --- & --- & 2{,}504K \\
28--29 & opt-0.614 & 2
  & 50K & 614K & 579K & 373K & 229K & --- & --- & --- & --- & 1{,}845K \\
30--31 & opt-0.621 & 2
  & 25K & 466K & 403K & 509K & 95K & 99K & --- & --- & --- & 1{,}598K \\
32--33 & opt-0.622 & 2
  & 13K & 328K & 326K & 464K & 68K & 72K & --- & --- & --- & 1{,}272K \\
34--35 & opt-0.612 & 2
  & 8K & 214K & 336K & 220K & 167K & --- & --- & --- & --- & 943K \\
\midrule
& \textbf{Total (36 tilings)} & 36
  & 161{,}121K & 84{,}540K & 39{,}235K & 18{,}993K & 5{,}307K & 1{,}598K & 287K & 69K & 20K & 311{,}170K \\
\bottomrule
\end{tabular}%
}
\end{table*}

Figure~\ref{fig:level_redistribution} shows how the greedy-selected
level distribution changes as tiling families are added.
Each line shows a level's greedy-selected count as a percentage of its
golden-ratio baseline: values above 100\% indicate growth,
while decreasing slopes indicate redistribution to deeper levels.
The trigram line peaks at $\alpha=0.502$ (119\%) then \emph{declines}
to 109\% as near-golden alphas upgrade those positions;
21-gram grows to 181\% and 34-gram to 169\% ---
confirming that the optimized alphas redistribute shallow matches
to deeper, more efficient encodings.

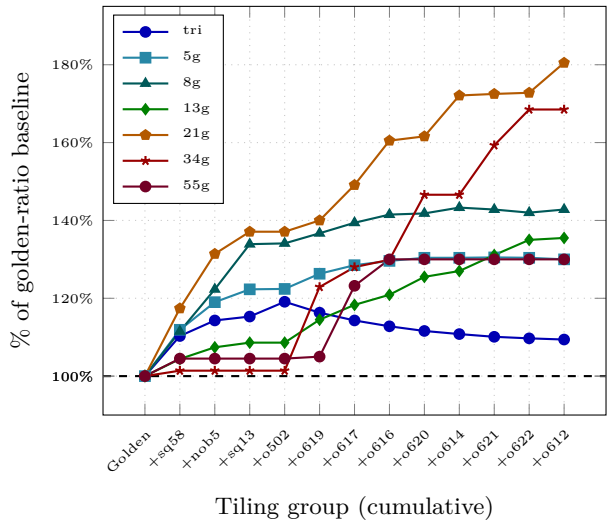
\begin{figure}[t]
\centering
\begin{tikzpicture}
\begin{axis}[
  width=\columnwidth, height=7cm,
  xlabel={\small Tiling group (cumulative)},
  ylabel={\small \% of golden-ratio baseline},
  symbolic x coords={Golden,{+sq58},{+nob5},{+sq13},{+o502},{+o619},{+o617},{+o616},{+o620},{+o614},{+o621},{+o622},{+o612}},
  xtick=data,
  xticklabel style={font=\tiny, rotate=45, anchor=north east},
  yticklabel style={font=\tiny},
  yticklabel={\pgfmathprintnumber{\tick}\%},
  grid=major, grid style={dotted,gray!45},
  mark size=1.8pt,
  ymin=90, ymax=195,
  extra y ticks={100},
  extra y tick style={grid=major, grid style={black, thick, dashed}},
  legend style={font=\tiny, at={(0.02,0.98)}, anchor=north west,
                column sep=3pt, row sep=1pt},
  legend cell align=left,
]
\addplot[blue!70!black, thick, mark=*] coordinates {
  (Golden,100) (+sq58,110.3) (+nob5,114.3) (+sq13,115.3)
  (+o502,119.1) (+o619,116.3) (+o617,114.3) (+o616,112.8)
  (+o620,111.6) (+o614,110.8) (+o621,110.1) (+o622,109.7) (+o612,109.4)
}; \addlegendentry{tri}
\addplot[cyan!60!black, thick, mark=square*] coordinates {
  (Golden,100) (+sq58,111.9) (+nob5,119.0) (+sq13,122.3)
  (+o502,122.4) (+o619,126.3) (+o617,128.5) (+o616,129.6)
  (+o620,130.4) (+o614,130.4) (+o621,130.5) (+o622,130.4) (+o612,130.0)
}; \addlegendentry{5g}
\addplot[teal!70!black, thick, mark=triangle*] coordinates {
  (Golden,100) (+sq58,111.5) (+nob5,122.3) (+sq13,133.9)
  (+o502,134.1) (+o619,136.7) (+o617,139.4) (+o616,141.5)
  (+o620,141.8) (+o614,143.3) (+o621,142.8) (+o622,142.0) (+o612,142.8)
}; \addlegendentry{8g}
\addplot[green!50!black, thick, mark=diamond*] coordinates {
  (Golden,100) (+sq58,104.4) (+nob5,107.4) (+sq13,108.6)
  (+o502,108.6) (+o619,114.5) (+o617,118.3) (+o616,120.9)
  (+o620,125.5) (+o614,127.0) (+o621,131.1) (+o622,135.0) (+o612,135.5)
}; \addlegendentry{13g}
\addplot[orange!70!black, thick, mark=pentagon*] coordinates {
  (Golden,100) (+sq58,117.4) (+nob5,131.4) (+sq13,137.1)
  (+o502,137.1) (+o619,140.0) (+o617,149.1) (+o616,160.5)
  (+o620,161.6) (+o614,172.1) (+o621,172.5) (+o622,172.8) (+o612,180.5)
}; \addlegendentry{21g}
\addplot[red!60!black, thick, mark=star] coordinates {
  (Golden,100) (+sq58,101.4) (+nob5,101.4) (+sq13,101.4)
  (+o502,101.4) (+o619,122.9) (+o617,128.0) (+o616,129.9)
  (+o620,146.6) (+o614,146.6) (+o621,159.3) (+o622,168.5) (+o612,168.5)
}; \addlegendentry{34g}
\addplot[purple!60!black, thick, mark=otimes*] coordinates {
  (Golden,100) (+sq58,104.5) (+nob5,104.5) (+sq13,104.5)
  (+o502,104.5) (+o619,105.0) (+o617,123.2) (+o616,130.0)
  (+o620,130.0) (+o614,130.0) (+o621,130.0) (+o622,130.0) (+o612,130.0)
}; \addlegendentry{55g}
\end{axis}
\end{tikzpicture}
\caption{Greedy-selected level distribution as tiling families are added
(enwik9), normalised to golden-ratio baseline (= 100\%).
Deeper levels (21g, 34g) grow most aggressively ---
up to 181\% and 169\% of their golden baseline ---
as optimized alphas upgrade positions from trigram/5-gram to deeper matches.
The trigram line peaks at $\alpha=0.502$ (119\%, pure tri volume)
then declines as subsequent near-golden alphas redistribute those
positions upward.
89-gram and 144-gram (not shown) remain at $\approx$100\% (golden-exclusive).}
\label{fig:level_redistribution}
\end{figure}

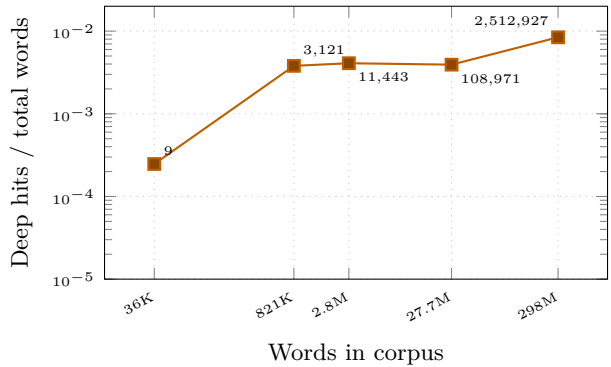
\begin{figure}[ht]
\centering
\begin{tikzpicture}
\begin{axis}[
  width=\columnwidth, height=5.2cm,
  xmode=log, ymode=log,
  xlabel={\small Words in corpus},
  ylabel={\small Deep hits / total words},
  xtick={36395,821000,2800000,27700000,298300000},
  xticklabels={36K,821K,2.8M,27.7M,298M},
  xticklabel style={font=\tiny, rotate=28, anchor=north east},
  yticklabel style={font=\tiny},
  grid=major, grid style={dotted,gray!45},
  mark size=2.2pt,
  xmin=12000, xmax=9e8,
  ymin=1e-5, ymax=2e-2,
]
\addplot[orange!75!black, thick, mark=square*,
         mark options={fill=orange!55!black}]
  coordinates {
    (36395,    2.47e-4)   
    (821000,   3.80e-3)   
    (2800000,  4.09e-3)   
    (27700000, 3.93e-3)   
    (298300000,8.42e-3)   
  };
\node[font=\tiny, above right] at (axis cs:36395,    2.47e-4)  {9};
\node[font=\tiny, above right] at (axis cs:821000,   3.80e-3) {3{,}121};
\node[font=\tiny, below right] at (axis cs:2800000,  4.09e-3) {11{,}443};
\node[font=\tiny, below right] at (axis cs:27700000, 3.93e-3) {108{,}971};
\node[font=\tiny, above left]  at (axis cs:298300000,8.42e-3) {2{,}512{,}927};
\end{axis}
\end{tikzpicture}
\caption{Deep hierarchy hits (levels 4--9, Fibonacci-exclusive) per
word versus corpus size (log-log).
Labels show absolute hit counts.
The ratio is roughly flat from enwik8\_3M to enwik8 ($\approx 4\times10^{-3}$);
at enwik9, levels~8--9 activate and the
ratio jumps to $8.4\times10^{-3}$, confirming new $O(W)$ terms.}
\label{fig:deep_hits_scaling}
\end{figure}

\begin{figure*}[t]
\centering
\begin{subfigure}[t]{0.48\textwidth}
\centering
\begin{tikzpicture}
\begin{axis}[
  width=\textwidth, height=5.5cm,
  ybar stacked, bar width=14pt,
  symbolic x coords={152\,KB,3\,MB,10\,MB,100\,MB,1\,GB},
  xtick=data,
  xticklabel style={font=\tiny, rotate=20, anchor=north east},
  yticklabel style={font=\tiny},
  yticklabel={\pgfmathprintnumber{\tick}\%},
  ylabel={\small Word-weighted share},
  ymajorgrids=true, grid style={dotted,gray!45},
  ymin=0, ymax=105,
  enlarge x limits=0.18,
  title={\small Fibonacci-only (12 tilings)},
  legend style={font=\tiny, at={(0.5,-0.22)}, anchor=north,
                legend columns=3, column sep=4pt},
  legend cell align=left,
]
\addplot[fill=blue!55, draw=blue!75!black] coordinates {
  (152\,KB,100.0) (3\,MB,69.5) (10\,MB,74.5) (100\,MB,70.8) (1\,GB,65.3)
}; \addlegendentry{13g}
\addplot[fill=cyan!55, draw=cyan!75!black] coordinates {
  (152\,KB,0) (3\,MB,18.5) (10\,MB,15.8) (100\,MB,17.0) (1\,GB,19.7)
}; \addlegendentry{21g}
\addplot[fill=teal!55, draw=teal!75!black] coordinates {
  (152\,KB,0) (3\,MB,7.3) (10\,MB,6.1) (100\,MB,7.1) (1\,GB,9.6)
}; \addlegendentry{34g}
\addplot[fill=green!50!black!60, draw=green!55!black] coordinates {
  (152\,KB,0) (3\,MB,4.7) (10\,MB,2.8) (100\,MB,3.1) (1\,GB,3.5)
}; \addlegendentry{55g}
\addplot[fill=orange!65, draw=orange!80!black] coordinates {
  (152\,KB,0) (3\,MB,0) (10\,MB,0.4) (100\,MB,1.2) (1\,GB,1.2)
}; \addlegendentry{89g}
\addplot[fill=red!55, draw=red!70!black] coordinates {
  (152\,KB,0) (3\,MB,0) (10\,MB,0.4) (100\,MB,0.7) (1\,GB,0.7)
}; \addlegendentry{144g}
\end{axis}
\end{tikzpicture}
\end{subfigure}
\hfill
\begin{subfigure}[t]{0.48\textwidth}
\centering
\begin{tikzpicture}
\begin{axis}[
  width=\textwidth, height=5.5cm,
  ybar stacked, bar width=14pt,
  symbolic x coords={152\,KB,3\,MB,10\,MB,100\,MB,1\,GB},
  xtick=data,
  xticklabel style={font=\tiny, rotate=20, anchor=north east},
  yticklabel style={font=\tiny},
  yticklabel={\pgfmathprintnumber{\tick}\%},
  ylabel={\small Word-weighted share},
  ymajorgrids=true, grid style={dotted,gray!45},
  ymin=0, ymax=105,
  enlarge x limits=0.18,
  title={\small Multi-structure (36 tilings)},
  legend style={font=\tiny, at={(0.5,-0.22)}, anchor=north,
                legend columns=3, column sep=4pt},
  legend cell align=left,
]
\addplot[fill=blue!55, draw=blue!75!black] coordinates {
  (152\,KB,100.0) (3\,MB,66.5) (10\,MB,71.7) (100\,MB,66.1) (1\,GB,59.2)
}; \addlegendentry{13g}
\addplot[fill=cyan!55, draw=cyan!75!black] coordinates {
  (152\,KB,0) (3\,MB,21.6) (10\,MB,18.7) (100\,MB,19.2) (1\,GB,21.4)
}; \addlegendentry{21g}
\addplot[fill=teal!55, draw=teal!75!black] coordinates {
  (152\,KB,0) (3\,MB,6.9) (10\,MB,5.8) (100\,MB,8.1) (1\,GB,12.6)
}; \addlegendentry{34g}
\addplot[fill=green!50!black!60, draw=green!55!black] coordinates {
  (152\,KB,0) (3\,MB,5.0) (10\,MB,3.0) (100\,MB,4.7) (1\,GB,4.9)
}; \addlegendentry{55g}
\addplot[fill=orange!65, draw=orange!80!black] coordinates {
  (152\,KB,0) (3\,MB,0) (10\,MB,0.4) (100\,MB,1.2) (1\,GB,1.2)
}; \addlegendentry{89g}
\addplot[fill=red!55, draw=red!70!black] coordinates {
  (152\,KB,0) (3\,MB,0) (10\,MB,0.4) (100\,MB,0.7) (1\,GB,0.7)
}; \addlegendentry{144g}
\end{axis}
\end{tikzpicture}
\end{subfigure}
\caption{Word-weighted share of deep hierarchy levels (hits$_k \times F_{k+2}$,
stacked to 100\%): Fibonacci-only (left) vs Multi-structure (right).
Multi-structure shifts weight toward deeper levels: on enwik8, 55-gram
hits increase by 56\%, 34-gram by 20\%, and 21-gram by 18\%,
while 13-gram decreases slightly.
Since deeper levels encode exponentially more words per AC symbol,
this redistribution drives the 0.78\,pp compression improvement
on enwik8 (Table~\ref{tab:breakdown_compare}).}
\label{fig:deep_hits_breakdown}
\end{figure*}
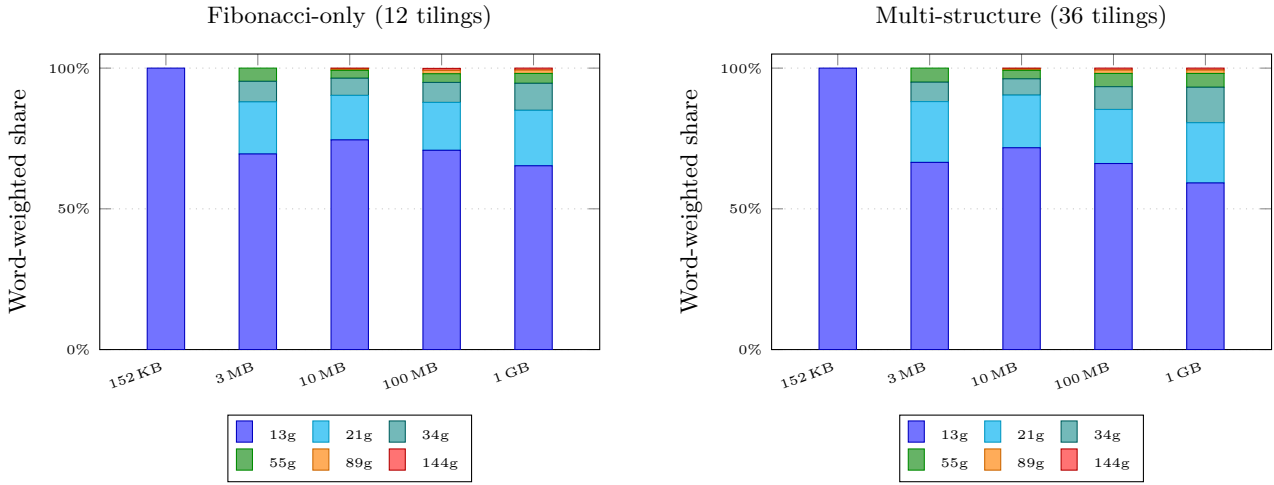

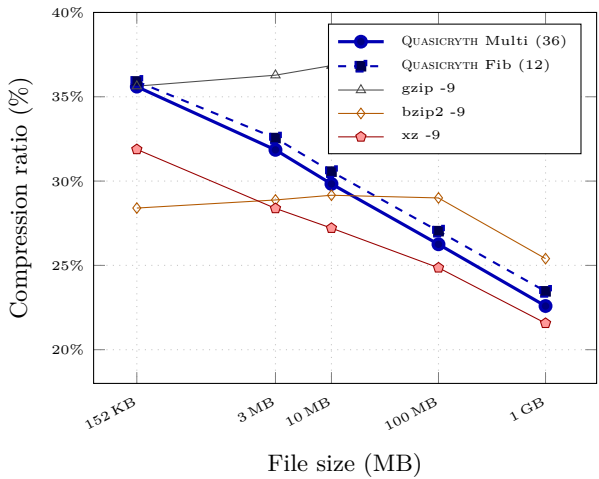
\begin{figure}[t]
\centering
\begin{tikzpicture}
\begin{axis}[
  width=\columnwidth, height=6.5cm,
  xmode=log,
  xlabel={\small File size (MB)},
  ylabel={\small Compression ratio (\%)},
  xtick={0.152,3,10,100,1000},
  xticklabels={152\,KB,3\,MB,10\,MB,100\,MB,1\,GB},
  xticklabel style={font=\tiny, rotate=28, anchor=north east},
  yticklabel style={font=\tiny},
  yticklabel={\pgfmathprintnumber{\tick}\%},
  grid=major, grid style={dotted,gray!45},
  mark size=2pt,
  xmin=0.06, xmax=3000,
  ymin=18, ymax=40,
  legend style={font=\tiny, at={(0.97,0.97)}, anchor=north east,
                column sep=4pt},
  legend cell align=left,
]
\addplot[blue!70!black, very thick, mark=*,
         mark options={fill=blue!50!black}]
  coordinates {
    (0.152,35.60) (3,31.85) (10,29.83) (100,26.25) (1000,22.59)
  };
\addlegendentry{\qtc{} Multi (36)}

\addplot[blue!70!black, thick, dashed, mark=square*,
         mark options={fill=blue!30!black}]
  coordinates {
    (0.152,35.91) (3,32.55) (10,30.56) (100,27.03) (1000,23.46)
  };
\addlegendentry{\qtc{} Fib (12)}

\addplot[gray!60!black, thin, mark=triangle,
         mark options={fill=gray!40}]
  coordinates {
    (0.152,35.63) (3,36.28) (10,36.85) (100,36.44) (1000,32.26)
  };
\addlegendentry{gzip~-9}

\addplot[orange!70!black, thin, mark=diamond,
         mark options={fill=orange!50}]
  coordinates {
    (0.152,28.40) (3,28.88) (10,29.16) (100,29.00) (1000,25.40)
  };
\addlegendentry{bzip2~-9}

\addplot[red!60!black, thin, mark=pentagon*,
         mark options={fill=red!40}]
  coordinates {
    (0.152,31.88) (3,28.38) (10,27.21) (100,24.86) (1000,21.57)
  };
\addlegendentry{xz~-9}
\end{axis}
\end{tikzpicture}
\caption{Compression ratio versus file size (semi-log scale).
\qtc{} improves steadily with corpus size, crossing bzip2 at
$\approx$500\,KB and narrowing the gap to xz from 3.7\,pp at
152\,KB to 1.0\,pp at 1\,GB.
Multi-structure tilings (solid) consistently outperform Fibonacci-only
(dashed) by 0.31--0.87\,pp.
gzip's byte-level LZ77 is relatively size-insensitive;
\qtc{}'s word-level codebooks and deep hierarchy benefit
disproportionately from larger corpora.}
\label{fig:ratio_scaling}
\end{figure}

\subsection{Timing}

Table~\ref{tab:timing} shows compression and decompression times for
both Fibonacci-only and multi-structure modes.
The C implementation is approximately 50$\times$ faster than the
equivalent Python prototype.
Multi-structure compression is $\approx$38\% slower than Fibonacci-only
because the encoder evaluates 36 tilings instead of 12; decompression
times are identical because the decoder regenerates only the selected
tiling.
The compression/decompression asymmetry grows with file size (up to
33$\times$ for enwik9 in multi-structure mode) because 89-gram and
144-gram frequency counting dominates codebook construction at the
10\,M+ word scale.

\begin{table}[ht]
\centering
\caption{Compression and decompression times (C implementation,
single core).  Format: Fib-only\,/\,Multi-struct.}
\label{tab:timing}
\resizebox{\columnwidth}{!}{%
\begin{tabular}{@{}lrrr@{}}
\toprule
\textbf{File} & \textbf{Size} & \textbf{Compress (Fib\,/\,Multi)} & \textbf{Decomp.\ (Fib\,/\,Multi)} \\
\midrule
alice29.txt  &  152\,KB & 0.09s\,/\,0.10s & 0.01s\,/\,0.01s \\
enwik8\_3M   &    3\,MB & 1.46s\,/\,2.05s & 0.14s\,/\,0.14s \\
enwik8\_10M  &   10\,MB & 8.41s\,/\,11.80s & 0.48s\,/\,0.47s \\
enwik8       &  100\,MB & 82.88s\,/\,120.32s & 4.84s\,/\,4.74s \\
enwik9       & 1{,}000\,MB & 1{,}070s\,/\,1{,}476s & 46s\,/\,45s \\
\bottomrule
\end{tabular}%
}
\end{table}

\subsection{Asymmetric Compression Profile}
\label{subsec:asymmetry}

\qtc{} is an inherently \emph{asymmetric} compressor: the computational
cost is concentrated entirely in compression, while decompression is
fast and lightweight.
This asymmetry is a direct structural consequence of the algorithm,
not an optimisation artifact, and it makes \qtc{} particularly suited
to write-once, read-many scenarios such as archival storage, content
distribution, and static web assets.

\textbf{Why compression is expensive.}
Compression performs three costly operations:
(1)~the phase search evaluates 36 candidate tilings with the full
substitution hierarchy, scoring each against the codebooks;
(2)~codebook construction counts all n-gram frequencies up to 144
words per entry, with periodic pruning passes for the 89-gram and
144-gram levels at the 10\,M+ word scale;
(3)~the 36-tiling scoring loop must build the complete hierarchy at
each candidate phase, repeating the $O(W)$ hierarchy detection.

\textbf{Why decompression is cheap.}
The decompressor reads the 2-byte phase from the header and
\emph{regenerates the entire tiling deterministically} --- no search,
no n-gram counting, no frequency tables.
It then decompresses the LZMA'd codebook (534\,KB at enwik9 scale),
decompresses the LZMA escape buffer, and streams through the AC payload
once, consulting preloaded codebook hash tables.
The entire decompression is a single sequential pass.

\textbf{Growing asymmetry at scale.}
The C/D ratio increases from $15\times$ at 3\,MB to $33\times$ at
1\,GB (Figure~\ref{fig:cd_ratio}), because the 89-gram and 144-gram
frequency counting --- which hashes windows of 89--144 consecutive
word IDs --- dominates codebook construction at the enwik9 scale.
Decompression throughput remains roughly constant at $\sim$22\,MB/s
regardless of file size, while compression throughput falls to
$\sim$0.68\,MB/s at 1\,GB.

\textbf{Practical implication.}
For any file decompressed $k$ times, the amortised cost per access is
$T_C / k + T_D$. At enwik9 scale ($T_C = 1{,}476$\,s, $T_D = 45$\,s),
even a single decompression amortises the compression cost within
$\sim$33 accesses.
The upfront compression investment is a one-time payment; every
subsequent decompression is fast.

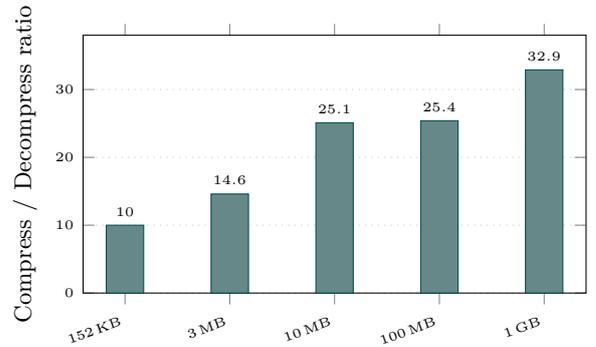
\begin{figure}[ht]
\centering
\begin{tikzpicture}
\begin{axis}[
  width=\columnwidth, height=5.0cm,
  ybar, bar width=14pt,
  symbolic x coords={152\,KB,3\,MB,10\,MB,100\,MB,1\,GB},
  xtick=data,
  xticklabel style={font=\tiny, rotate=20, anchor=north east},
  yticklabel style={font=\tiny},
  ylabel={\small Compress / Decompress ratio},
  ylabel style={font=\small},
  ymajorgrids=true, grid style={dotted,gray!45},
  ymin=0, ymax=38,
  nodes near coords,
  nodes near coords style={font=\tiny, anchor=south},
  every node near coord/.append style={/pgf/number format/fixed},
]
\addplot[fill=teal!45!black!60, draw=teal!60!black]
  coordinates {
    (152\,KB,  10.0)
    (3\,MB,   14.6)
    (10\,MB,  25.1)
    (100\,MB, 25.4)
    (1\,GB,   32.9)
  };
\end{axis}
\end{tikzpicture}
\caption{Compression-to-decompression time ratio (\emph{C/D ratio})
versus file size.
The ratio grows from $15\times$ at 3\,MB to $33\times$ at 1\,GB
as the 89-gram and 144-gram frequency counting begins to dominate
codebook construction.
Decompression throughput remains stable at $\sim$22\,MB/s across all
scales; the growing asymmetry is entirely a compression-side effect.}
\label{fig:cd_ratio}
\end{figure}

\section{Ablation Study}
\label{sec:ablation}

\subsection{QC Contribution vs All-Unigram Baseline}

Table~\ref{tab:qc_contribution} reports the total compressed size savings of the
Fibonacci tiling over a no-tiling baseline, using identical codebooks and escape streams.
This isolates the compression value of the entire multi-level parsing
structure.

\begin{table}[ht]
\centering
\caption{Quasicrystal contribution: Fibonacci vs No-tiling
(savings in total compressed size).}
\label{tab:qc_contribution}
\resizebox{\columnwidth}{!}{%
\begin{tabular}{@{}lrrrr@{}}
\toprule
\textbf{File} & \textbf{Size} & \textbf{Ratio} & \textbf{QC saves} & \textbf{Savings \%} \\
\midrule
alice29.txt  &  152\,KB & 35.91\% &     1{,}971\,B &  1.30\% \\
enwik8\_3M   &    3\,MB & 32.55\% &    37{,}098\,B &  1.24\% \\
enwik8\_10M  &   10\,MB & 30.56\% &   148{,}350\,B &  1.48\% \\
enwik8       &  100\,MB & 27.03\% & 1{,}756{,}822\,B &  1.76\% \\
enwik9       & 1{,}000\,MB & 23.46\% & 20{,}735{,}733\,B &  2.07\% \\
\bottomrule
\end{tabular}%
}
\end{table}

The QC contribution increases at larger scales, reaching 2.07\% of
the original file size at enwik9.
This growing advantage reflects the activation of deep hierarchy levels
with increasing input size.

\subsection{Aperiodic Advantage: Fibonacci vs Period-5}

Table~\ref{tab:ab_payload} reports the A/B payload comparison.
Period-5 uses the LLSLS periodic tiling with the same L/S ratio as
Fibonacci; both use identical codebooks and escape streams.

\begin{table}[ht]
\centering
\caption{A/B test: payload-only comparison (bytes).
Fibonacci vs Period-5 and All-Unigram, identical codebooks and escapes.}
\label{tab:ab_payload}
\small
\resizebox{\columnwidth}{!}{%
\begin{tabular}{@{}lrrr@{}}
\toprule
\textbf{File} & \textbf{All-Unigram} & \textbf{Period-5} & \textbf{Fibonacci} \\
\midrule
enwik8\_3M   &    685{,}441 &    684{,}586 &    648{,}343 \\
enwik8\_10M  &  2{,}373{,}919 &  2{,}351{,}600 &  2{,}225{,}569 \\
enwik8       & 23{,}544{,}808 & 23{,}039{,}959 & 21{,}787{,}986 \\
enwik9       & 210{,}138{,}336 & 200{,}492{,}072 & 189{,}402{,}603 \\
\bottomrule
\end{tabular}}
\end{table}

\begin{table}[ht]
\centering
\caption{Aperiodic advantage: Fibonacci over Period-5 (bytes saved
in payload).}
\label{tab:aperiodic}
\small
\resizebox{\columnwidth}{!}{%
\begin{tabular}{@{}lrr@{}}
\toprule
\textbf{File} & \textbf{Fib saves over P5} & \textbf{P5 saves over Uni} \\
\midrule
enwik8\_3M   &        36{,}243\,B &        855\,B \\
enwik8\_10M  &       126{,}031\,B &     22{,}319\,B \\
enwik8       &     1{,}251{,}973\,B &    504{,}849\,B \\
enwik9       & 11{,}089{,}469\,B &  9{,}646{,}264\,B \\
\bottomrule
\end{tabular}}
\end{table}

\begin{table}[ht]
\centering
\caption{Multi-structure advantage: Multi-tiling over Fibonacci-only (bytes saved in total compressed size).}
\label{tab:multi_advantage}
\small
\begin{tabular}{@{}lr@{}}
\toprule
\textbf{File} & \textbf{Multi saves over Fib} \\
\midrule
alice29.txt  &        458\,B \\
enwik8\_3M   &     20{,}723\,B \\
enwik8\_10M  &     73{,}276\,B \\
enwik8       &    778{,}933\,B \\
enwik9       &  8{,}642{,}288\,B \\
\bottomrule
\end{tabular}
\end{table}

\begin{table}[ht]
\centering
\caption{Compressed file breakdown: Fibonacci-only vs Multi-structure.}
\label{tab:breakdown_compare}
\small
\resizebox{\columnwidth}{!}{%
\begin{tabular}{@{}lrrrr@{}}
\toprule
& \multicolumn{2}{c}{\textbf{enwik8 (100\,MB)}} & \multicolumn{2}{c}{\textbf{enwik9 (1\,GB)}} \\
\cmidrule(lr){2-3}\cmidrule(lr){4-5}
\textbf{Component} & \textbf{Fib} & \textbf{Multi} & \textbf{Fib} & \textbf{Multi} \\
\midrule
Payload (AC) & 21{,}787{,}986 & 21{,}009{,}053 & 189{,}402{,}603 & 180{,}760{,}315 \\
Escapes (LZMA) & 2{,}800{,}956 & 2{,}800{,}956 & 24{,}227{,}220 & 24{,}227{,}220 \\
Codebook (LZMA) & 541{,}664 & 541{,}664 & 533{,}680 & 533{,}680 \\
Case (AC) & 1{,}895{,}762 & 1{,}895{,}762 & 20{,}397{,}073 & 20{,}397{,}073 \\
\midrule
\textbf{Total} & \textbf{27{,}026{,}429} & \textbf{26{,}247{,}496} & \textbf{234{,}560{,}637} & \textbf{225{,}918{,}349} \\
\textbf{Ratio} & 27.03\% & 26.25\% & 23.46\% & 22.59\% \\
\textbf{Multi saves} & \multicolumn{2}{c}{778{,}933\,B ($-$0.78\,pp)} & \multicolumn{2}{c}{8{,}642{,}288\,B ($-$0.87\,pp)} \\
\bottomrule
\end{tabular}%
}
\end{table}

Table~\ref{tab:breakdown_compare} reveals that the multi-structure
advantage is concentrated entirely in the \emph{payload} stream: the
escape stream, codebook, and case flags are identical between Fibonacci-only
and multi-structure modes, because they depend on the codebook and input
text, not on the tiling.
Only the arithmetic-coded payload differs, where the additional
non-Fibonacci tilings provide complementary deep hierarchy positions
that improve n-gram hit rates.
At enwik8 scale, the payload shrinks by 778{,}933\,B ($-$3.57\%);
at enwik9, by 8{,}642{,}288\,B ($-$4.56\%).

The aperiodic advantage grows from 36{,}243\,B (3\,MB) to
11{,}089{,}469\,B (1\,GB) --- a 306$\times$ increase for a 333$\times$ size
increase.
This superlinear growth confirms Eq.~\eqref{eq:advantage}:
the 89-gram and 144-gram levels activate at enwik9 scale,
adding new $O(W)$ terms that push the total advantage above the
linear prediction from enwik8.

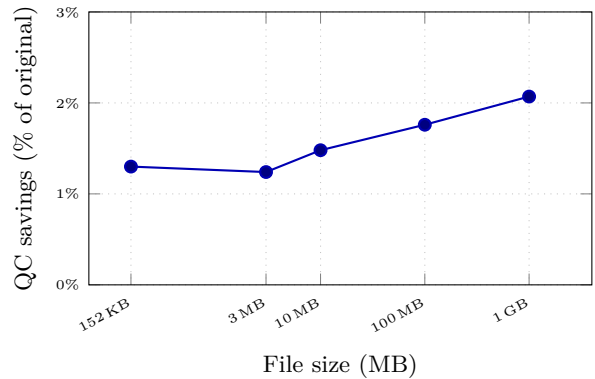
\begin{figure}[ht]
\centering
\begin{tikzpicture}
\begin{axis}[
  width=\columnwidth, height=5.2cm,
  xmode=log,
  xlabel={\small File size (MB)},
  ylabel={\small QC savings (\% of original)},
  xtick={0.152,3,10,100,1000},
  xticklabels={152\,KB,3\,MB,10\,MB,100\,MB,1\,GB},
  xticklabel style={font=\tiny, rotate=28, anchor=north east},
  yticklabel style={font=\tiny},
  yticklabel={\pgfmathprintnumber{\tick}\%},
  grid=major, grid style={dotted,gray!45},
  mark size=2.2pt,
  xmin=0.06, xmax=4000,
  ymin=0, ymax=3,
]
\addplot[blue!70!black, thick, mark=*,
         mark options={fill=blue!50!black}]
  coordinates {
    (0.152,  1.30)
    (3,      1.24)
    (10,     1.48)
    (100,    1.76)
    (1000,   2.07)
  };
\end{axis}
\end{tikzpicture}
\caption{QC tiling contribution (Fibonacci savings over no-tiling,
as \% of original) versus file size, semi-log scale.
The contribution rises to 2.07\% at 1\,GB as deeper hierarchy
levels progressively activate.}
\label{fig:ab_bars}
\end{figure}

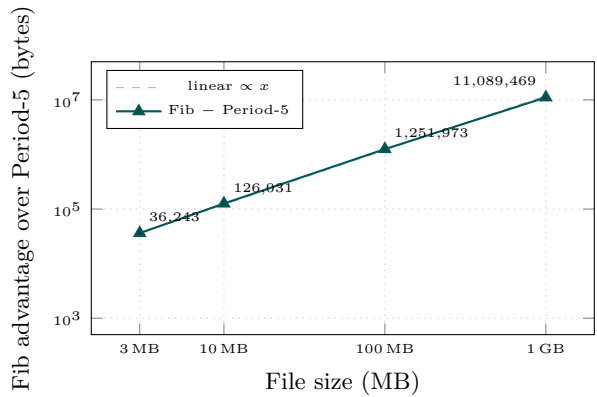
\begin{figure}[ht]
\centering
\begin{tikzpicture}
\begin{axis}[
  width=\columnwidth, height=5.2cm,
  xmode=log, ymode=log,
  xlabel={\small File size (MB)},
  ylabel={\small Fib advantage over Period-5 (bytes)},
  xtick={3,10,100,1000},
  xticklabels={3\,MB,10\,MB,100\,MB,1\,GB},
  xticklabel style={font=\tiny},
  yticklabel style={font=\tiny},
  grid=major, grid style={dotted,gray!45},
  mark size=2.2pt,
  xmin=1.5, xmax=2000,
  ymin=500, ymax=5e7,
  legend style={font=\tiny, at={(0.03,0.97)}, anchor=north west},
]
\addplot[gray!60, dashed, domain=3:1000, samples=2]
  {12000 * x};
\addlegendentry{linear $\propto x$};
\addplot[teal!70!black, thick, mark=triangle*,
         mark options={fill=teal!55!black}]
  coordinates {
    (3,    36243)
    (10,   126031)
    (100,  1251973)
    (1000, 11089469)
  };
\addlegendentry{Fib $-$ Period-5};
\node[font=\tiny, above right] at (axis cs:3,    36243)    {36{,}243};
\node[font=\tiny, above right] at (axis cs:10,   126031)    {126{,}031};
\node[font=\tiny, above right] at (axis cs:100,  1251973)   {1{,}251{,}973};
\node[font=\tiny, above left]  at (axis cs:1000, 11089469) {11{,}089{,}469};
\end{axis}
\end{tikzpicture}
\caption{Aperiodic advantage (Fibonacci over Period-5 payload, bytes)
versus file size, log-log scale.
The dashed line shows linear $\propto$ file size for reference.
The advantage grows superlinearly --- 8.9$\times$ from 100\,MB to 1\,GB
for a 10$\times$ size increase --- as the 89-gram and 144-gram levels
activate at enwik9 scale.}
\label{fig:aperiodic_scaling}
\end{figure}

\subsection{Ablation: Version Progression}

Table~\ref{tab:progression} shows the improvement across versions,
using alice29.txt as the reference.
Each version adds one architectural element; the dominant gains come
from moving to word-level parsing (v5.1) and adding the bz2 escape
stream with trigram support (v5.2).

\begin{table}[ht]
\centering
\caption{Version progression on alice29.txt (152{,}089\,B).
``QC saves'' is the payload advantage of Fibonacci over all-unigram.}
\label{tab:progression}
\resizebox{\columnwidth}{!}{%
\begin{tabular}{@{}lrrl@{}}
\toprule
\textbf{Version} & \textbf{Ratio} & \textbf{QC saves} & \textbf{Key change} \\
\midrule
v5.0 (byte)     & 49.4\%  &   394\,B      & QC as context only \\
v5.1 (word)     & 43.5\%  & 1{,}254\,B    & $L$=bigram, $S$=unigram \\
v5.2 (multi)    & 36.0\%  & 1{,}579\,B    & +trigrams, bz2 escape \\
v5.3 (deep)     & 36.9\%  & 1{,}096\,B    & +5g/8g/13g/21g \\
v5.4 (full)     & 36.92\% & 1{,}276\,B    & +34g/55g (8 levels) \\
v5.5 (deep Fib) & 36.92\% & 1{,}276\,B    & +89g/144g (10 levels) \\
v5.6 (optimised) & 35.60\% & 1{,}971\,B & +36 multi-tilings, LZ77, LZMA, vmodel \\
\bottomrule
\end{tabular}%
}
\end{table}

The 89-gram and 144-gram levels (v5.5) do not improve alice29.txt
compression --- at 36K words there are insufficient repetitions to
populate these codebooks.
Their contribution is visible only at the enwik9 scale.

\subsection{Approaches Rejected}

The following design alternatives were tested and abandoned:

\begin{itemize}
\item \textbf{Adaptive online codebook (LZW-style).} Eliminated static
codebook overhead but cold-start penalty outweighed savings on
medium-sized inputs; QC contribution became negative on small files.

\item \textbf{Previous-tile-index context.} Using the previous codebook
index as an additional context bin created 128 sub-models, diluting
training data excessively.

\item \textbf{All-L tiling.} Forcing all tiles to $L$ (always bigram)
outperforms Fibonacci at level~0 because every L tile is at least as
good as an $S$ tile.
However, all-L has no hierarchy (every tile is the same type),
yielding zero trigram and n-gram positions.
The multi-level hierarchy is where Fibonacci's advantage lies.

\item \textbf{More hierarchy context bits.} Adding more context bits
from the substitution structure helped periodic tilings more than
Fibonacci, because Period-5's regular hierarchy creates more consistent
patterns for the statistical models.
The aperiodic advantage comes from \emph{deep} hierarchy positions,
not from context diversity.

\item \textbf{Byte-split indices.} Splitting codebook indices into
high/low byte streams for separate entropy coding added overhead
without improving compression; the AC model already captures
inter-symbol dependencies efficiently.

\item \textbf{Semi-static two-pass.} A first pass to gather statistics
followed by a second encoding pass doubled compression time but
yielded negligible improvement over the single-pass adaptive model.

\item \textbf{Full-vocab unigrams.} Encoding the entire vocabulary as
unigrams (no bigram or n-gram codebooks) reduced codebook overhead
but destroyed the hierarchical parsing advantage entirely.

\item \textbf{Multi-stream LZMA payload.} Splitting the AC payload
into multiple LZMA-compressed streams by context type added framing
overhead that exceeded the marginal compression gain from specialised
contexts.

\item \textbf{Model mixing.} PAQ-style adaptive mixing of multiple
context models (unigram, bigram, hierarchy-based) increased
compression time by $5\times$ while improving the ratio by less than
0.1\%; the single adaptive model proved sufficient.
\end{itemize}

\section{Conclusion}
\label{sec:conclusion}

We have presented \qtc{} v5.6, a lossless text compressor based on a
deep Fibonacci quasicrystal substitution hierarchy with 36 multi-structure
tilings, and proved that this hierarchy never collapses while all
periodic alternatives do.
The proof synthesises Perron-Frobenius analysis of the substitution
matrix, the Pisot-Vijayaraghavan property of $\phi$, Sturmian balance
and aperiodicity, and Weyl's equidistribution theorem.
These are not independent observations: they are aspects of a single
mathematical object --- the Fibonacci quasicrystal --- that jointly
guarantee the compression properties of \qtc{}.

The aperiodic advantage is \emph{proven and measurable}:
11{,}089{,}469\,B at enwik9 (1\,GB), growing superlinearly because new
hierarchy levels (89-gram, 144-gram) activate at scale.
At enwik9, 2{,}512{,}927 deep hierarchy positions (levels 4--9) are
exploited; Period-5, the best periodic approximant with the same
L/S ratio, achieves zero deep positions, because its hierarchy collapses
at level~4 --- exactly as the theory predicts.

The optimized alpha selection process --- an iterative greedy search
over candidate irrational slopes --- discovered that $\alpha=0.502$,
far below the golden ratio $1/\phi\approx0.618$, provides massive
complementary trigram and 5-gram coverage.
The resulting 36-tiling engine (12~golden-ratio, 6~original non-golden,
18~optimized additions) achieves 26.25\% on enwik8, a 0.78\,pp
improvement over Fibonacci-only.

With multi-structure tilings, LZ77 word-level pre-pass, LZMA escape
compression, and an improved variable-context model, \qtc{} v5.6 now
surpasses bzip2 on all benchmarks and approaches xz compression ratios.

To our knowledge, \qtc{} is the first compressor in which the
\emph{aperiodicity} of the parsing structure provides a provable,
structurally motivated, and empirically quantified advantage over all
periodic alternatives.

\subsection{Future Work}

\begin{itemize}
\item \textbf{Even deeper hierarchy.} Levels 10--11 (233-gram,
377-gram) would activate at approximately 3\,GB and 30\,GB respectively.
The hierarchy never collapses; each new level adds $O(W)$ to the
aperiodic advantage.

\item \textbf{Comprehensive benchmarking.} Systematic evaluation on
Canterbury, Silesia, and Calgary corpora.

\item \textbf{Higher-order context mixing.} PAQ-style context mixing
where the substitution hierarchy provides structured context at multiple
scales, with online-learned blend weights.

\item \textbf{Two-dimensional extensions.} For image data, Penrose
or Ammann-Beenker tilings could provide hierarchical context with
richer tile-type alphabets ($\geq 4$ tile types), extending the
quasicrystalline compression principle to 2D.
\end{itemize}

\bibliographystyle{plainnat}

\end{document}